\documentclass[conference,a4paper]{IEEEtran}
\usepackage{amsmath, amssymb, amsbsy,cite}
\usepackage{mathtools}
\usepackage{amsthm,multirow,color,amsfonts}
\usepackage{tabulary}
\usepackage{subfigure}
\usepackage{graphicx}
\usepackage{setspace}
\usepackage{enumerate}
\usepackage{bbm,bm}
\usepackage{algorithm}
\usepackage{algpseudocode}
\usepackage{comment}
\usepackage{setspace}
\usepackage{enumitem,kantlipsum}
\newif\ifshort
\shortfalse
\newtheorem{theo}{Theorem}
\newtheorem{defi}{Definition}

\newtheorem{cor}{Corollary}

\newcommand{\off}[1]{}
\newcommand{\ton}[1]{\textcolor[rgb]{0.00,0.00,0.00}{#1}}

\newcommand{\mm}[1]{\textcolor[rgb]{0.00,0.00,0.00}{#1}}
\newcommand{\RTT}{{\sf RTT}}
\newcommand{\ACK}{{\sf ACK}}
\newcommand{\NACK}{{\sf NACK}}
\newcommand{\Dm}{D_{\sf mean}}
\newcommand{\DM}{D_{\sf max}}
\newcommand{\ov}{\bar{o}}

\title{Adaptive Causal Network Coding with Feedback\off{ for Delay and Throughput Guarantees}\vspace{-1mm}}
 \author{%
   \IEEEauthorblockA{Alejandro Cohen\IEEEauthorrefmark{1},
                     Derya Malak\IEEEauthorrefmark{1},
                     Vered Bar Bracha\IEEEauthorrefmark{2},
                     and Muriel M{\'e}dard\IEEEauthorrefmark{1}}
   \IEEEauthorblockA{\IEEEauthorrefmark{1}%
                     Research Laboratory of Electronics, MIT, Cambridge, MA, USA,
                     \{cohenale, deryam, medard\}@mit.edu}
   \IEEEauthorblockA{\IEEEauthorrefmark{2}%
                     Intel Corporation,
                     vered.bar.bracha@intel.com}\thanks{This research was supported in part by the Intel Corporation and by DARPA. Patent application submitted: no. 62/853,090.}
 \vspace{-9mm}}
\IEEEoverridecommandlockouts
\begin{document}
\maketitle
%%%%%%%%%%%%%%%%%%%%%%%%%%%%%%%%%%%%%%%%%%%%%%%%%%%%%%%%%%%%%%%
\begin{abstract}
%%%%%%%%%%%%%%%%%%%%%%%%%%%%%%%%%%%%%%%%%%%%%%%%%%%%%%%%%%%%%%%
We propose a novel adaptive and causal random linear network coding (AC-RLNC) algorithm with forward error correction (FEC) for a point-to-point communication channel with delayed feedback. \mm{AC-RLNC is adaptive to the channel condition, that the algorithm estimates, and is  causal, as coding depends on the particular erasure realizations, as reflected in the feedback acknowledgments. Specifically,} the proposed model can learn the erasure pattern of the channel via feedback acknowledgments, and adaptively adjust its retransmission rates using a priori and posteriori algorithms. By those adjustments, AC-RLNC achieves the desired delay and throughput, and enables transmission with zero error probability. We upper bound the throughput and the mean and maximum in order delivery delay of AC-RLNC, and  prove that for the point to point communication channel in the non-asymptotic regime the proposed  code  may achieve more than 90\% of the channel capacity. To upper bound the throughput we utilize the minimum Bhattacharyya distance for the AC-RLNC code. We validate those results via simulations. We contrast  the performance of AC-RLNC with the one of  selective repeat (SR)-ARQ, which is causal but not adaptive, and is a posteriori. Via a study on experimentally obtained commercial traces, we demonstrate that a protocol based on AC-RLNC can, vis-\`a-vis SR-ARQ, double the throughput gains, and triple the gain in terms of mean in order delivery delay when the channel is bursty. Furthermore, the difference between the maximum and mean in order delivery delay is much smaller than that of SR-ARQ. Closing the delay gap along with boosting the throughput is very promising for enabling ultra-reliable low-latency communications (URLLC) applications. %We validate the performance of data delivery under the traces of Intel, which confirms the correctness of our analysis.\off{ We think this will be promising for URLLC applications.}
\end{abstract}
%%%%%%%%%%%%%%%%%%%%%%%%%%%%%%%%%%%%%%%%%%%%%%%%%%%%%%%%%%%%%%%
\section{Introduction}
%%%%%%%%%%%%%%%%%%%%%%%%%%%%%%%%%%%%%%%%%%%%%%%%%%%%%%%%%%%%%%%
Classical information theory problems consider the very large blocklength regime to achieve the desired communication rates. On the other hand, in streaming communications, there are real-time constraints on the transmission which requires low delays. Such low delays cannot be achieved using very large blocklengths; therefore the classical approaches do not provide the desired throughput-delay tradeoffs. To alleviate the problem of large in order delivery delays, different forward error correction (FEC) techniques for packet-level coding have been contemplated \cite{shokrollahi2006raptor,luby2002lt,cloud2015coded,GuoShiCaiMed2013}. We discuss their benefits and shortcomings in this section.

The challenges associated with packet-level coding are multifold, which are due to feedback, real-time delivery, and congestion. The situation becomes deteriorated when there are round-trip time (RTT) fluctuations, and/or erasure bursts due to variations in the channel state. When we consider these critical issues, along with the real-time transmission constraints, it becomes challenging to operate the system at a desired point on the achievable throughput and in order delivery delay regime. As the notions of throughput and delay are inseparable, it is essential to jointly optimize coding for throughput and delay. To that end, in this paper, we provide an adaptive and causal code construction for packet scheduling. To ensure desired rates the sender use FEC. To provide low delay, the sender tracks channel variations and learns channel states from feedback in order to adapt the code rate to the channel conditions. \mm{In other words, the solution is adaptive since AC-RLNC can track the channel condition, and it is also causal because coding is reactive  to the feedback acknowledgments (i.e., the sender adjusts the retransmission rates of the code using a priori and posteriori algorithms).}

\begin{figure*}
    \centering
    \subfigure[]{\includegraphics[scale=0.3]{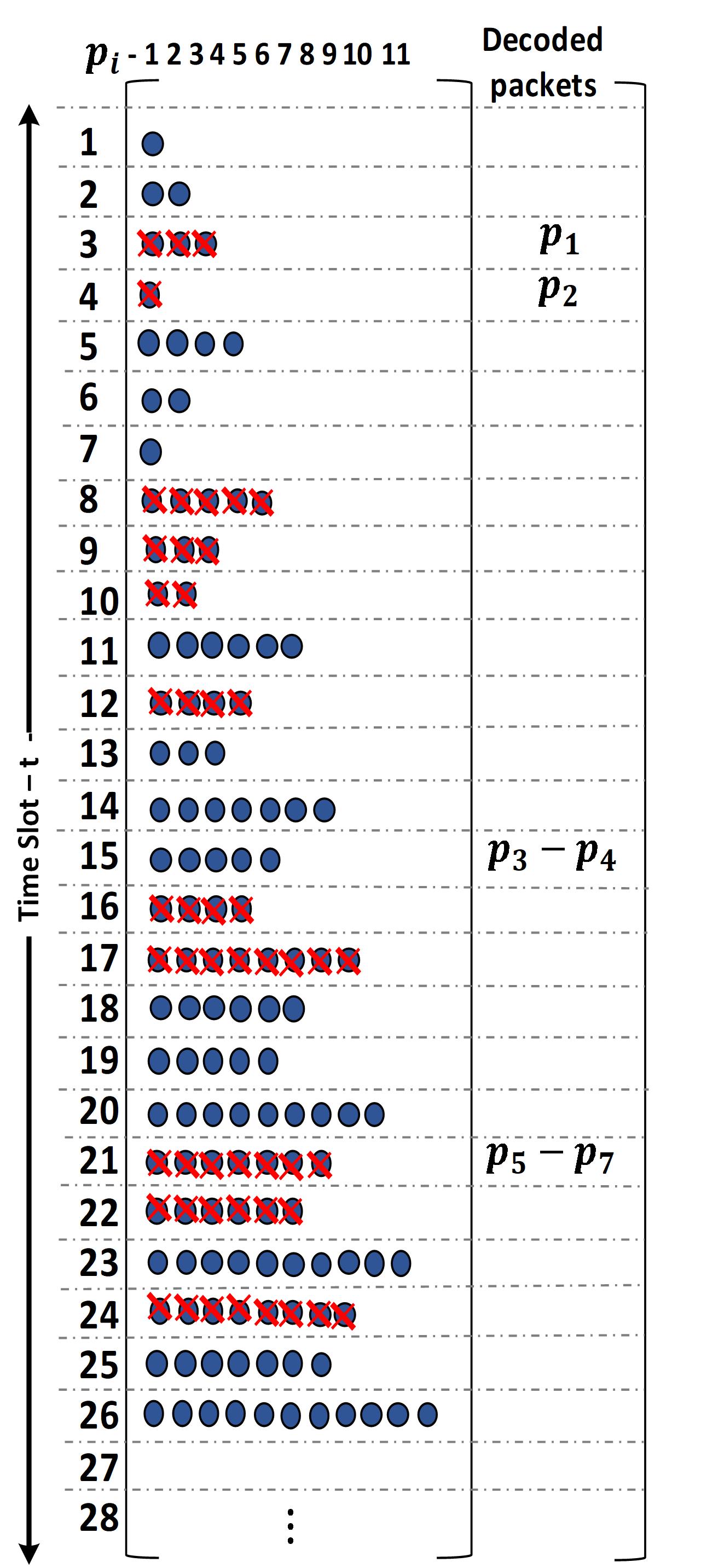}}
    \subfigure[]{\includegraphics[scale=0.3]{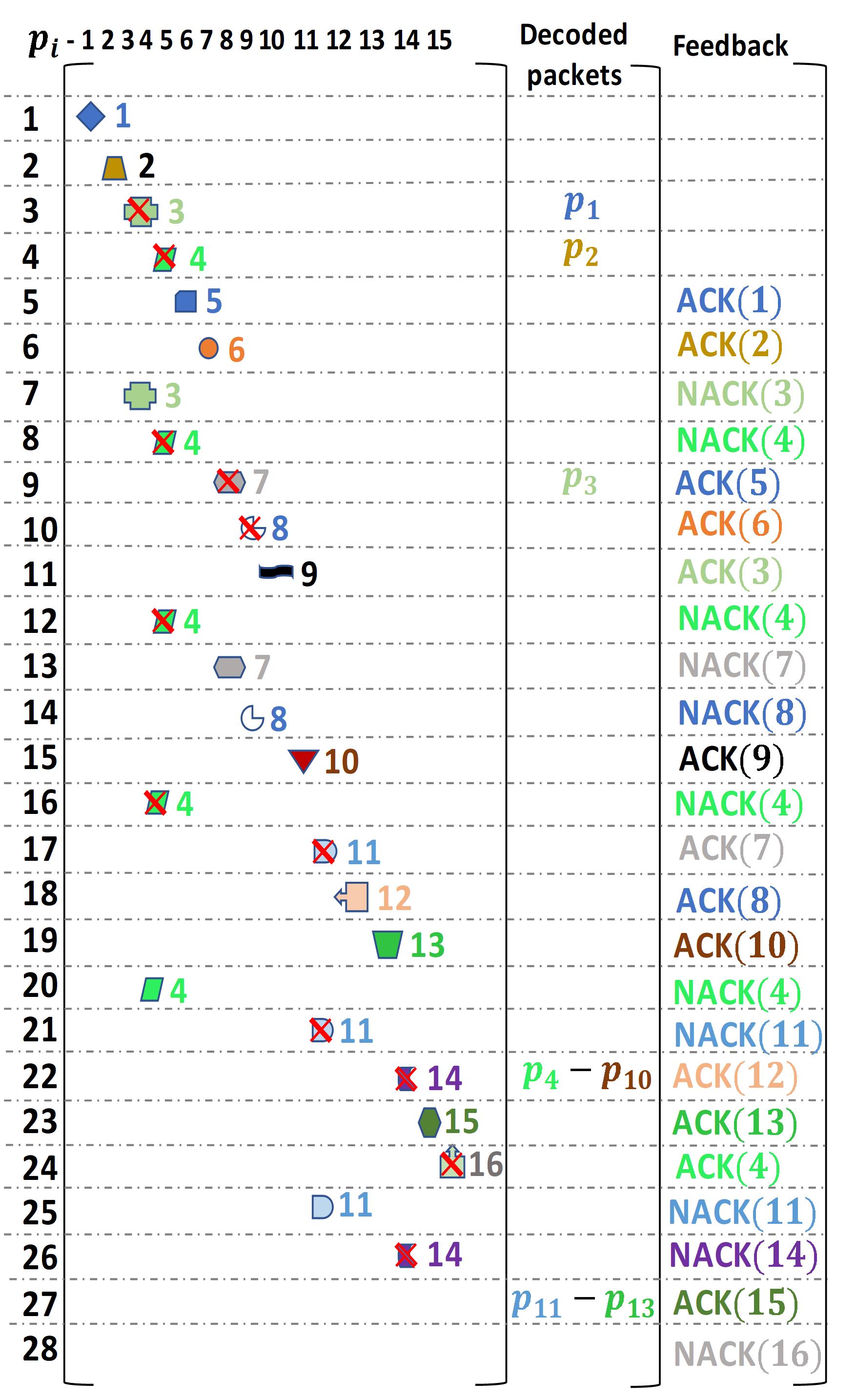}}
    \subfigure[]{\includegraphics[scale=0.3]{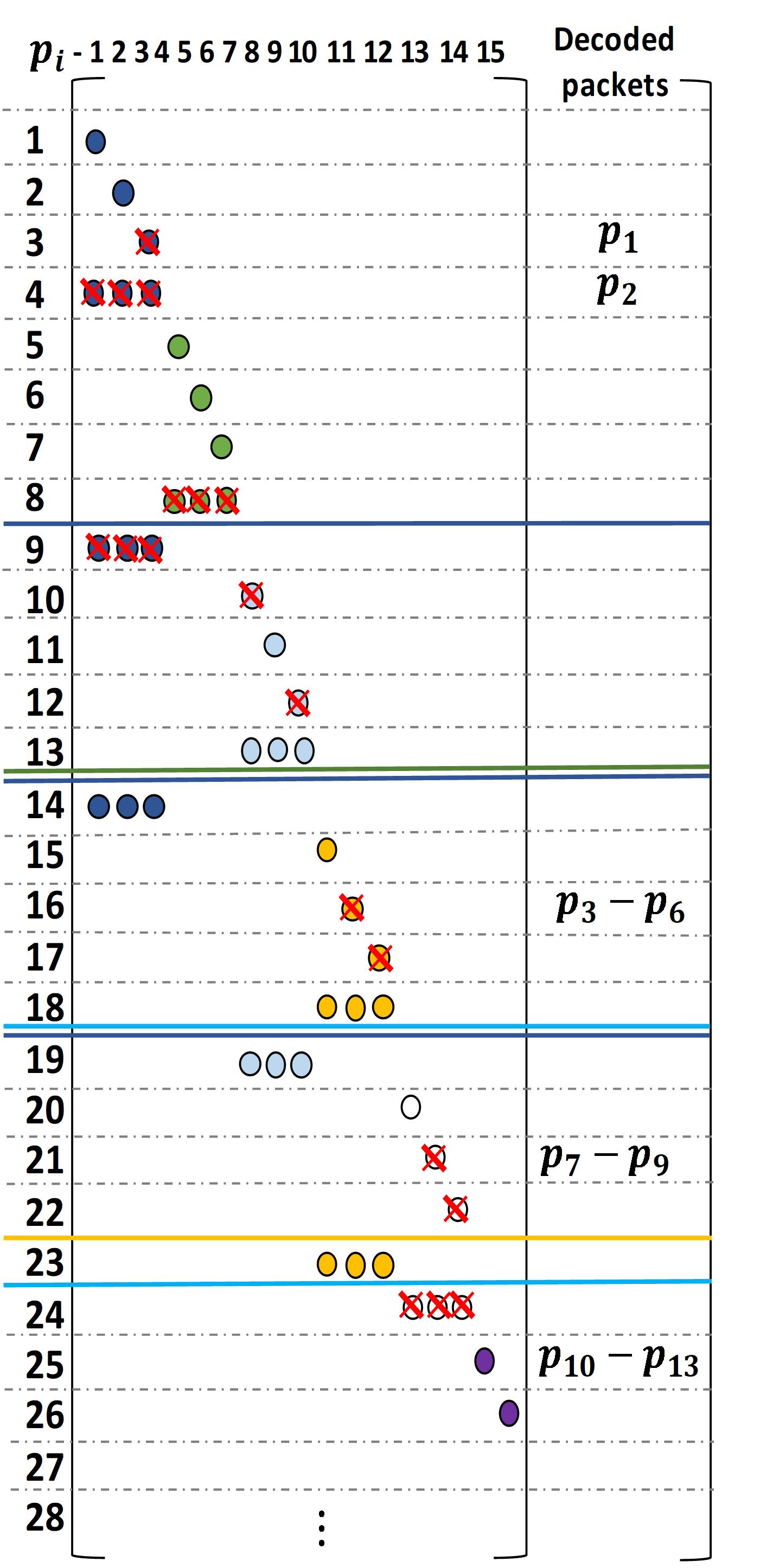}}
    \subfigure[]{\includegraphics[scale=0.3]{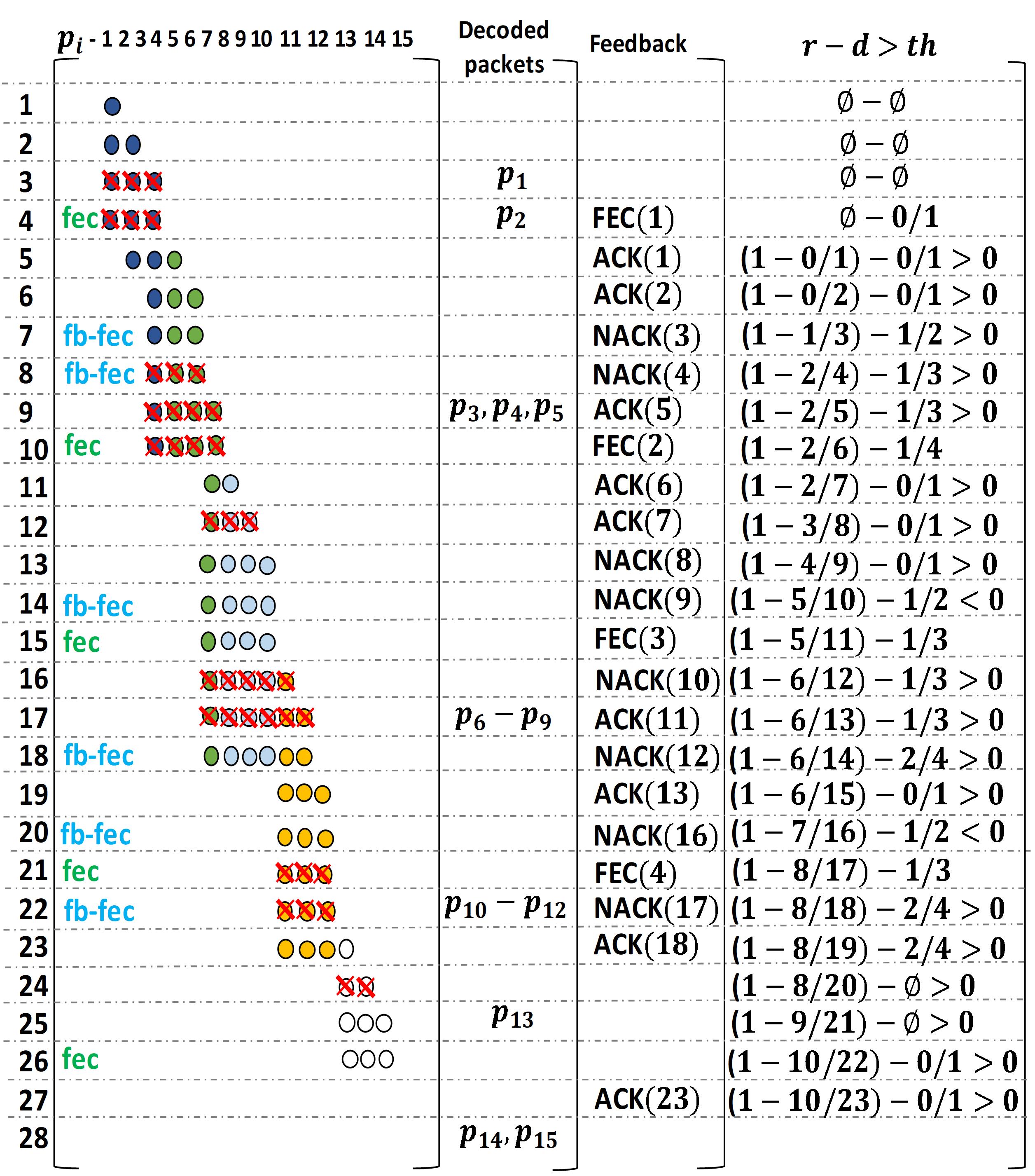}}
    \caption{Coding matrices assuming the same loss pattern and RTT of 4 time-slots. (a) Time-Invariant Streaming Code \cite{joshi2012playback},\cite{joshi2014effect}; \mm{(b) SR-ARQ \cite{weldon1982improved,anagnostou1986performance,ausavapattanakun2007analysis};} (d) Systematic Code with Feedback \cite{cloud2015coded}; (d) AC-RLNC with Feedback (Our Scheme). In the codes presented in (a) and (c), if the actual rate of the channel is lower than the code or FEC rate, respectively, the in order delay increases. If the actual rate of the channel is higher than the code or FEC rate, respectively, the throughput decreases. \mm{In the SR-ARQ protocol presented in (b), it is possible to obtain the maximum throughput of the channel, yet, at the expense of a long in order delay.} In the AC-RLNC presented in (d) to manage this tradeoff, first, a priori FEC is transmitted every RTT period according to the actual rate of the channel. Second, posteriori the sender adaptively and causally decides if to add new information packet to the RLNC or send additional DoF according to the rate of the channel. In the last column, $d$ is updated after the sender computes $r-d$ in a way that $m_d=m_d+1$, $a_d=a_d$. Therefore, the equation $r-d>0$ becomes correct after the update. In Appendix \ref{App:adaptivecausalRLNC}, we provide a detailed description of the example (d).}
    \label{fig:foobar}
     \vspace{-3mm}
\end{figure*}

%%%%%%%%%%%%%%%%%%%%%%%%%%%%%%%%%%%%%%%%%%%%%%%%%%%%%%%%%%%%%%%
\subsection{Related Work}
%%%%%%%%%%%%%%%%%%%%%%%%%%%%%%%%%%%%%%%%%%%%%%%%%%%%%%%%%%%%%%%
Different coding techniques have been proposed to correct erasures in wireless channels such as chunk codes, rateless erasure codes or fountain codes (e.g., LT codes \cite{luby2002lt} and Raptor codes \cite{shokrollahi2006raptor}), systematic codes \cite{cloud2015coded}, and streaming codes \cite{malak2019tiny}. \mm{However, although part of the solutions proposed in the literature are reactive according to the feedback acknowledgments (i.e. causal as given for example in \cite{cloud2015coded}), none of those solutions is tracking the varying channel condition and the rate (i.e. not adaptive). The proposed techniques utilize several different types of feedback models \cite{bertsekas1992data,malak2019tiny}. Next, we will elaborate on the coding approaches that are the closest to ours.} While fountain codes are capacity achieving, and have efficient encoding and decoding algorithms, they are not suitable for streaming because the decoding delay is proportional to the size of the data \cite{LubMitShoSpiSte1997}. In the case of block codes, the receiver has to wait till the end of the block to be able to start decoding. However, the required block length for the code to achieve a desired reliability performance can be very high due to the asymptotic nature of information theoretical results. Furthermore, the block length is fixed and set in advance for the application of interest, and therefore block codes are not causal. To mitigate this problem and reduce the in order delivery in wireless systems, authors in \cite{KarLei2014} have proposed a low delay streaming code scheme which is not adaptive. Similarly, convolutional codes can reduce the decoding time, and a general construction for complete maximum distance profile convolutional codes has been implemented in \cite{lieb2018complete}.

Error control protocols such as Automatic Repeat reQuest (ARQ), and hybrid ARQ (HARQ) have been incorporated into 5G mobile networks \cite{3gpp2017study} in order to increase the performance of wireless technologies such as HSPA, WiMAX and LTE \cite{khosravirad2017analysis}. While these repetition-based protocols provide desired performance when the feedback is perfect, the streaming quality might degrade significantly when the feedback is not reliable or delayed. This may cause extra latency, which is not desired in delay-sensitive applications. Recently, for streaming, systematic codes, which are coded generalizations of selective repeat ARQ (SR-ARQ) have been proposed in \cite{cloud2015coded}. Adaptive coded ARQ model with cumulative feedback have been developed in \cite{malak2019tiny}.
\off{In addition, systematic codes have been proposed in \cite{cloud2015coded}, which are coded generalizations of selective repeat ARQ (SR-ARQ), and the adaptive coded ARQ model with cumulative feedback as in \cite{malak2019tiny}.}

Using FEC, in order delivery delay over packet erasure channels can be reduced \cite{KarLei2014}, and the performance of SR-ARQ protocols can be boosted \cite{cloud2015coded}. Delay bounds for convolutional codes have been provided in \cite{GuoShiCaiMed2013}, \cite{TomFitLucPedSee2014}. Packet dropping to reduce the playback delay of streaming over an erasure channel has been investigated in \cite{joshi2012playback}, \cite{joshi2014effect}, \cite{joshi2016efficient}. Delay-optimal codes without feedback for burst erasure channels, and the decoding delay of codes for more general erasure models have been analyzed in \cite{Martinian2004}. Transmission with delay constraints has been considered in \cite{adams2014delay}, by combining the PHY and NET layer aspects, where bit level FEC is performed at the PHY layer, and Random Linear Network Coding (RLNC) is performed at the NET layer. To prevent packet losses in the presence of interference and large RTT, a network coded TCP solution has been proposed in \cite{kim2012network}. In Fig. \ref{fig:foobar}, we refer to \mm{three} of the coding approaches that are the closest to ours. However, these alternatives are not \mm{adaptive. Furthermore, the time-invariant streaming code in (a)} is not causal, and the \mm{ SR-ARQ in (b) and the systematic code with feedback in (c)} do not adapt the code rate to varying channel conditions, unlike our AC-RLNC proposed solution shown in \mm{(d) that tracks the channel conditions and the rate}. In Sec. \ref{back}, we provide more details on those codes.

Delay and throughput gains of coding in unreliable networks have been discussed in \cite{ErOzMedAh2008}. The delay advantage of coding in packet erasure networks has been studied in \cite{DikDimHoEff2014}. A capacity-achieving coding scheme for unicast or multicast over lossy packet networks has been proposed in \cite{LunMedKoeEff2008}, where intermediate nodes perform recoding and send out coded packets formed from random linear combinations of previously received packets. Joint optimization of coding for delay for single hop erasure channels, and packet scheduling for a broadcast channel with multiple receivers having different delay sensitivities has been considered \cite{zeng2012joint}. The single hop model has been later generalized to multiple-hops \cite{MSME19}. There also exist work applying FEC and network coding on URLLC applications  \cite{malak2018throughput,cloud2015coded,anand2018resource}.

Current approaches address some of the challenges such as reducing the in order delivery delay to provide the desired reliability-delay trade-offs. However, the coding in general is done in a deterministic manner. This approach deteriorates the performance when the channel is bursty, RTT fluctuates, or real-time transmission constraints are imposed.

\begin{comment}
%References sent by Intel
%We can use these for multipath
%Multipath (RLNC) references \cite{gabriel2018multipath}, forward error correction (FEC) \cite{tesema2018layer},
\cite{8424259}
\cite{domazetovic2016performance}, \cite{sun2013ieee}
\cite{du2015network}: they approximate the end-to-end probability distribution of a network by combining the parallel/tandem channel reduction with the structure of flows over the network (You might want to check this paper Alejandro.)
\end{comment}

%%%%%%%%%%%%%%%%%%%%%%%%%%%%%%%%%%%%%%%%%%%%%%%%%%%%%%%%%%%%%%%
\subsection{Contributions}
%%%%%%%%%%%%%%%%%%%%%%%%%%%%%%%%%%%%%%%%%%%%%%%%%%%%%%%%%%%%%%%
We propose a novel adaptive causal coding scheme with FEC for a point-to-point communication channel with delayed feedback. The proposed model can track the erasure pattern of the channel, and adaptively adjusts its retransmission rates a priori and posteriori based on the channel quality (the erasure burst pattern) and the feedback acknowledgments. \off{The algorithm suggested contains two retransmition mechanisms. The first, a priori estimate the channel by the delayed feedback. The second, posteriori decides the number of retransmissions.}

%Analysis
We provide analytical results for the adaptive code suggested herein. Specifically, we upper bound the mean and max in order delivery delays\footnote{In this paper, we only consider the in order delivery delay given in Definition \ref{dife_delay}.} and the throughput. We prove that for the point to point communication in the non-asymptotic regime the code proposed may achieve more than 90\% of the channel capacity, and observe significant gains in the mean and maximum in order delivery delays, compared to SR-ARQ. The gains become more apparent when the channel is more bursty and RTT is high. The main distinctions of this approach from standard approaches like SR-ARQ \cite{3gpp2017study}, or other FEC schemes \cite{KarLei2014} which are very sensitive to the fluctuations in the channel quality (bursts) are that the proposed adaptive causal coding scheme is more robust to the burst erasures, and the round trip delay, and can handle in order delivery delay requirements. The algorithm parameters can be chosen to bound the delay and throughput of AC-RLNC, and achieve the desired delay-throughput tradeoffs. AC-RLNC also enables transmission with zero error probability. Zero error capacity has been considered in the literature \cite{shannon1956zero,gallager1965simple,lovasz1979shannon}, and for the finite block length regime \cite{sahai2008block,polyanskiy2011feedback}.

%Simulations
%performance comparison with other works
The simulation results for the implementation of AC-RLNC demonstrate the robustness of the algorithm. They also show that in real wireless scenarios, in addition to the improvement in throughput, the gap between the mean in order delay and the maximum in order delay is very small unlike the one in SR-ARQ where the growth rate of the gap is higher. In consistent with these, we validate the performance of data delivery of our algorithm via experimental simulations under the traces of Intel. \off{We explain the details about the Intel measurements in Sect. \ref{simulation}. The analytical results presented have a good agreement with the simulation results.}

The rest of the paper is structured as follows. In Sect. \ref{back}, we provide some definitions and metrics. In Sect. \ref{sys}, we formally describe the system model and problem formulation. In Sect. \ref{algo}, we present the adaptive causal network coding algorithm, and in Sect. \ref{bounds}, we provide the analytical results. In Sect. \ref{simulation}, we describe an experimental study and simulations exemplifying the performance of the proposed method. Finally, we conclude the paper in Sect. \ref{conc} with possible future directions.

%\vspace{-0.4cm}
%%%%%%%%%%%%%%%%%%%%%%%%%%%%%%%%%%%%%%%%%%%%%%%%%%%%%%%%%%%%%%%
\section{Background}\label{back}
We provide background on the important metrics we study.

\begin{defi}{In order delivery delay, $D$.}\label{dife_delay} This is the difference between the time an information packet is first transmitted in a coded packet by the sender and the time that the same information packet is decoded, in order at the receiver, and successfully acknowledged
\cite{zeng2012joint}.
\end{defi}

The in order delivery delay also includes the decoding delay of packets at the receiver. We assume that decoding is via Gaussian elimination\footnote{Given a large enough field $\mathbb{F}_z$ (when $z$ is the field size), the receiver can decode a generation of $k$ packets with high probability, through Gaussian elimination performed on the linear system formed on any $k$ coded packets \cite{ho2006random}. % so that we get sufficient number of independent combinations of packets.
}.

The mean in order delivery delay, $\Dm$, is the average value of $D$. This metric is of interest in reducing the overall completion delay of all packets such as in file download.

The maximum in order delivery delay, $\DM$, is the maximum value of $D$ among all the information packets in the stream. This metric is of interest in reducing the maximum inter-arrival time between any two packets with new information, which might be critical for real-time applications, e.g., video streaming, conference calls, or distributed systems in which real-time decisions are taken according to information received from another source in the system.

We use random linear network coding (RLNC) as a forward error correction (FEC) technique to control the erasures over the unreliable forward channel. The benefit of RLNC is that the sender does not need to retransmit the same packet, and the receiver only has to collect enough degrees of freedom (DoF) to be able to decode the packets within the generation window. The number of information packets contained in an RLNC coded packet should be determined in accordance with the DoF needed by the receiver. We detail how we adaptively determine the effective window size in Sect. \ref{algo}. We assume the packet overhead due to the RLNC header and the decoding time using Gaussian elimination are negligible \cite{patterson2014and,chou2003practical}.

The successful delivery of each packet is acknowledged by the receiver on a per packet basis. We assume that the forward channel is unreliable and the reverse channel (feedback) is reliable. Therefore, the feedback acknowledgements are in the form of ACK's or NACK's which are successfully delivered after one $\RTT$.

\begin{defi}{Throughput, $\eta$.}
This is the total amount of information (in bits/second) delivered, in order at the receiver in $n$ transmissions over the forward channel. The normalized throughput is the total amount of information delivered, in order at the receiver divided by $n$ and the size of the packets.
\end{defi}

%\tof{how does the throughput compare to the capacity of the channel? Can you please explain this Alejandro}

%How are the delay and throughput guaranteed?
The concepts of in order delivery delay and throughput are not separable. The higher the throughput, i.e., the total amount of information we want to deliver using the forward channel, the higher the in order delay, causing the tradeoff between delay and throughput. The FEC mechanism along with the feedback determines the performance of the protocol. There exist work on characterizing this tradeoff, such as in fixed nodes \cite{gupta2000capacity}, mobile \cite{grossglauser2001mobility}, ad hoc mobile networks \cite{neely2005capacity}, and fixed and mobile ad hoc networks \cite{gamal2004throughput}.
However, to the best of our knowledge, the exact characterization and optimization of the tradeoff between in order delay and throughput might not be possible (or tractable) \cite{malak2019tiny}. Via the proposed AC-RLNC scheme, one of the main objectives of the current paper is to provide guarantees (see Sects. \ref{algo}, \ref{simulation}) for in order delay and throughput, instead of the explicit characterization of the tradeoff.
\begin{figure}
    %{\includegraphics[width =\columnwidth]{figures/window_coding_v2}}
    {\includegraphics[trim=0cm 0.48cm 0cm 0.1cm,clip,scale=0.5]{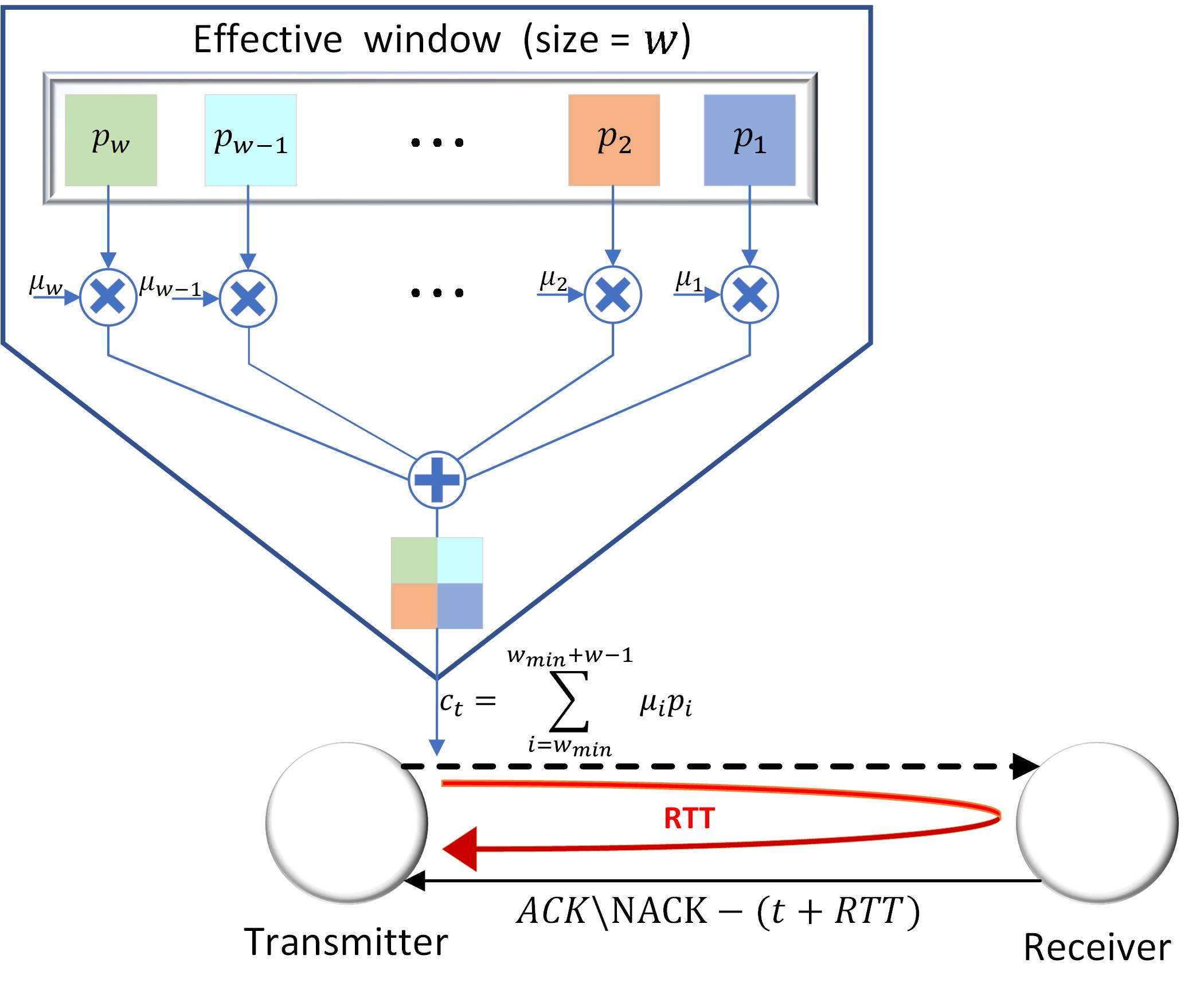}}
	\caption{System model and encoding process of the coded RLNC combination.
	\mm{The adaptive causal encoding process and the effective window size $w$ are detailed in Sec. \ref{algo}.} In this example, for simplicity of notation $w_{\min}=1$. }
	\label{fig:window_coding_v2}
	\vspace{-0.4cm}
\end{figure}

In Fig. \ref{fig:foobar}, we illustrate three different codes. The rows and columns of a coding matrix denote the time slot $t$ and information packet $i$ indices, respectively. At a given time slot (corresponding to a particular row of the coding matrix), a coded packet, which is random linear combination of information packets $p_i$ is transmitted, as shown in Fig. \ref{fig:foobar}-\mm{(a), (c) and (d)}. \mm{In Fig. \ref{fig:foobar}-(b), the packets are not coded. Hence, each individual information packet $p_i$ is transmitted at a given time slot.} The lost packets are shown with crossed marks. Appended to the right of each coding matrix, we have the indices of decoded packets\off{(in order)} and the corresponding time slots.

In Fig. \ref{fig:foobar}-(a), we provide the coding matrix for a time-invariant streaming code. This scheme does not involve feedback from the receiver. Hence, the sender keeps transmitting random linear combinations of the information packets by enlarging the size of transmission window over time. \mm{In Fig. \ref{fig:foobar}-(b), we show the coding matrix for SR-ARQ protocol. In this protocol, the objective is to only retransmit the information packets that were erased, then re-order the packets at the receiver. That is, the sender at each time slot decides whether to send a new packet of information or retransmit the erased packet according to the feedback acknowledgments. For example, at time slot $11$, the sender receives an ACK for information packet $3$ that was transmitted at time slot $7$. Since the sender does not have the feedback for the packets transmitted between time slots $8$ and $11$, i.e. information packets $4-8$ as shown in the coding matrix, it decides to send the next information packet that is not transmitted yet, i.e. information packet $9$. At time slot $12$, the sender receives the NACK of erased information packet $4$. Therefore, it decides to resend information packet $4$. This protocol utilizes a sliding window to manage the throughput-delay tradeoff as given in \cite{bertsekas1992data}. In Fig. \ref{fig:foobar}-(c), we show the coding matrix for} a systematic code with feedback. This case, unlike streaming, is a generation-based scheme. %The coding buckets act as the Head of Line (HOL) generations in the most generation based scheme.
The solid colored lines show the receipt of the feedback at the sender, and their colors correspond to the given generation. For example, the feedback (NACK) for the first generation is received at time slot $8$. Then the sender knows that it needs to transmit another DoF. At time slot $13$ the sender receives the feedback for the first and second generations, and hence the sender knows the second generation is successfully transmitted, and $1$ more DoFs is required for the receiver to be able to successfully decode the first and the second generations in order.\off{Since the packets have to be decoded in order, the sender should prioritize the generations according to the feedback.At time slot $13$ the sender receives the feedback for the second and third generations, and hence the sender knows the second generation is successfully transmitted, and $2$ more DoFs is required for the receiver to be able to successfully decode the third generation. Since the packets have to be decoded in order, the sender should prioritize the generations according to the feedback.} Finally, in Fig. \ref{fig:foobar}-(d), we present AC-RLNC with feedback (our scheme). In the proposed protocol, the sender tracks rate and missing DoF according to the feedback acknowledgments. The feedback is per coded packet, i.e. per slot. It then adds DoFs according to the rate of the channel (every RTT period). Furthermore, the sender adaptively decides if to add new information packet to the RLNC or send additional DoF according to the rate of the channel and the rate of the DoFs. %We bound the number of information packets in RLNC to limit the number of missing DoFs.
The details of how we choose the coding matrix will be clear in Sect. \ref{algo}.

\begin{table}[t!]%\small
\centering
\begin{tabular}{|l|l|l|}
\hline
{\bf Parameter} & {\bf Definition} \\
\hline
$t$ & time slot index  \\
$M$ & number of information packets\\
$p_i,\, i\in [1,M]$ & information packets\\
$c_t$ & RLNC to transmit at time slot $t$ \\
$\mu_i \in \mathbb{F}_z$ &  random coefficients\\
$e$ & total number of erasures in $[1,t]$ \\%transmitted packets (update according NACK)   \\
$\Dm$, $\DM$ & the mean and maximum in order\\
& delivery delay of packet\\
$p_e$ & erasure probability \\
$m_d$ & number of DoFs needed by the\\ & receiver to decode $c_t$ \\%(update according NACK) \\
$a_d$ & DoF added to $c_t$  \\%(update according NACK and FEC) \\
$d$ & rate of DoF ($m_d/a_d$) \\
$r$ & rate of the channel\\
%& (e.g. for BEC, $1-p_e=1-e/t$) \\
$th$ & throughput-delay tradeoff parameter\\
$r-d>th$ & retransmission criterion \\ %adjusted to choose the throughput-delay trade-off
$\RTT=k+1$ & round-trip time \\
$w_{\min}$ & index of the first information packet\\
& in $c_t$\\
EW& end window of $k$ new packets \\
$m$ & number of FEC to add per window \\
$\ov$ & maximum number of information \\
& packets allowed to overlap\\
$w\in\{1,\hdots,\ov\}$ & effective window size \\
E$\ov$W & end overlap window of maximum\\
&  new packets\\
\hline
\end{tabular}
\vspace{0.2cm}
\label{fig:table}
\caption{AC-RLNC algorithm: symbol definitions.}
\vspace{-0.8cm}
\end{table}

The scheme presented in this paper is related to fountain codes \cite{mackay2005fountain,sejdinovic2009expanding}, for which the code design very much depends on the choice of degree distributions. %\tof{Do you need to optimize degree distributions here?}
Different from fountain codes which do not exhibit a fixed rate, AC-RLNC is not rateless since the window size and hence the number of information packets to be combined at any time slot is bounded. It is also different from a block code with fixed dimension or message length. Unlike systematic codes, AC-RLNC do not restrict the coding into different generations. Furthermore, it is causal\footnote{Coding schemes can be non-causal. Block codes make for erasures that have not yet been seen, rather than reaching to them in the way, say a rateless code does.}.

%%%%%%%%%%%%%%%%%%%%%%%%%%%%%%%%%%%%%%%%%%%%%%%%%%%%%%%%%%%%%%%
\section{System Model and Problem Formulation}\label{sys}
%%%%%%%%%%%%%%%%%%%%%%%%%%%%%%%%%%%%%%%%%%%%%%%%%%%%%%%%%%%%%%%
We consider an adaptive, causal, real-time, slotted, point-to-point communication model with feedback (i.e. AC-RLNC) for low-latency constraints. Fig. \ref{fig:window_coding_v2} shows the system model. We consider the case where erasures may occur over the forward channel. In each time slot $t$ the sender transmits a coded packet $c_t$ over the forward channel to the receiver. To simplify the technical aspects and focus on the key methods, we assume that the feedback channel is noiseless. The receiver acknowledges the sender for each coded packet transmitted over the feedback channel. Denote by $t_{p}$ the maximum propagation delay over any channel, and by $t_{d}=|c_t|/r$ the transmission delay of the packet, where $|c_t|$ is the size of each coded packet in bits and $r$ denote the rate of the channel in bits/second. Since the sender transmits one coded packet per time slot, $t_d$ is also the duration of a time-slot. We assume the size of the acknowledgements is negligible compared to the packet size. Hence, the round-trip time is
\begin{align}
\RTT = t_{d} + 2t_{p}.
\end{align}
For each $t$-th coded packet it transmits, the sender receives a reliable $\ACK(t)$ or $\NACK(t)$ after $\RTT$.

For the forward channel we consider two types of channels. The first case is a binary erasure channel (BEC), with an i.i.d. erasure probability of $\epsilon$ per slot. Let $n$ denote the total number of transmission slots. Thus, on average, $n(1-\epsilon)$ slots are not erased and are available to the receiver. The second is a Gilbert-Elliott (GE) channel with erasures. This channel is a binary-state Markov process with good ($G$) and bad ($B$) states. It introduces bursts into the channel, and hence isolates erasures \cite{gilbert1960capacity,sadeghi2008finite,malak2019tiny}. Let ${\bf P}$ be the probability transition matrix of the GE channel, which is given as
\begin{align}
    {\bf P}=\begin{bmatrix}
    1-q & q\\
    s & 1-s\\
    \end{bmatrix},
\end{align}
where the first (second) row represents the transition probabilities from the good (bad) state. The stationary distribution satisfies $\pi_G=\frac{s}{q+s}$, and $\pi_B=1-\pi_G$. Let $\epsilon_G=0$ and $\epsilon_B=1$ be the erasure rates at the corresponding states. The average erasure rate is $\epsilon=\pi_B$. Note that $1/s$ is the average erasure burst. Hence, burst erasures occur when $s$ is low.

With the proposed AC-RLNC scheme, with parameters $r$, $n$ and $\RTT$, our goal is to minimize the in order delivery delay, $D$, and maximize the throughput, $\eta$.

%%%%%%%%%%%%%%%%%%%%%%%%%%%%%%%%%%%%%%%%%%%%%%%%%%%%%%%%%%%%%%%
\section{Adaptive Coding Algorithm}\label{algo}
%%%%%%%%%%%%%%%%%%%%%%%%%%%%%%%%%%%%%%%%%%%%%%%%%%%%%%%%%%%%%%%
In this section, we detail the AC-RLNC given in Algorithm \ref{causalRLNCalgo}. AC-RLNC differs from SR-ARQ in terms of the structure of the feedback and the retransmission criterion, which is mainly affected by the feedback, the window size and the forward channel conditions. Namely, the sender tracks the channel rate and the DoF rate at the receiver via the feedback acknowledgments, and selects if to add a new information packet to the next coded RLNC packet it sends\footnote{\label{note2}\mm{We consider a packet-level communications model. The bit-level communication can be considered as a special case of our model.}}. Fig. \ref{fig:window_coding_v2} shows the system model and the adaptive causal encoding process of the AC-RLNC protocol with an effective window size $w$. The symbol definitions and an example realization of AC-RLNC are provided in Table \ref{fig:table} and Fig. \ref{fig:foobar}-(d), respectively. In Appendix \ref{App:adaptivecausalRLNC} we provide a detailed description of the example with the coding matrix as given in Fig. \ref{fig:foobar}-(d). The main components of the packet level protocol are described next.

\begin{algorithm}[t!]\small
\begin{algorithmic}
\While{DoF($c_t$)$>0$}
\State $t = t+1$
\State Update $d=\frac{m_d}{a_d}$ according to the known encoded packets
\If{no feedback}
    \If{EW}
        \State Transmit the same\off{\footnotemark[\ref{note1}]} RLNC $m$ times: $a_d=a_d+m$
    \Else
        \State Add new $p_i$ packet to the RLNC and transmit
    \EndIf
\ElsIf{feedback NACK}
    \State $e=e+1$
    \State Update $m_d$ according to the known encoded packets
        \If{$r-d > th$}
            \If{not EW}
            \State Add new $p_i$ packet to the RLNC and transmit
            \Else
            \State Transmit the same RLNC $m$ times: $a_d=a_d+m$
            \EndIf
         \Else  %$r-d < th$
         \State Transmit the same RLNC $a_d=a_d+1$
           \If{EW}
           \State Transmit the same RLNC $m$ times: $a_d=a_d+m$
           \EndIf
        \EndIf
\Else % (feedback ACK)
      \If{EW}
         \State Transmit the same RLNC $m$ times: $a_d=a_d+m$
      \EndIf

      \If{$r-d < th$}
         \State Transmit the same RLNC
      \Else
         \State Add new $p_i$ packet to the RLNC and transmit
      \EndIf
\EndIf
\State Eliminate the seen packets from the RLNC
\If{DoF($c_t$)$>\ov$}
   \State transmit the same RLNC until DoF($c_t$)$=0$
\EndIf
\EndWhile
 \caption{Adaptive causal RLNC for packet scheduling.\label{causalRLNCalgo}}
\end{algorithmic}
\end{algorithm}

%%%
\paragraph{Tracking Channel Behavior} \label{ch_behavior}
In order to adapt the causal algorithm the sender estimates the actual channel behavior (i.e. erasure probability and its variance, and the burst pattern) using the feedback acknowledgements. To focus on the key methods, we concentrate on the BEC here. However, the same methods apply to the GE channel or more general channel models.
First, the sender counts the actual number of erasures $e$ at each time slot $t$. However, this number is computed based on the acknowledgments corresponding to time $t-\RTT$, and hence is an estimate. If there was no delay in sender's estimate, it would achieve the capacity. The sender also keeps an estimate of the standard deviation of erasures due to the errors estimation caused by the round-trip delay.
The probability of erasure at slot $t$, $p_e=e/(t-\RTT)$, is the fraction of erasures over the time interval $[1,t-\RTT]$. Hence, the sender can compute the channel rate as $r=1-p_e$, and the standard deviation for BEC as $\sqrt{p_e(1-p_e)}$. Similarly, in the GE channel, we can estimate the actual burst pattern of the channel to adapt the algorithm.

%%%
\paragraph{Window Structure and Different Generations}
The sliding window structure is determined by the $\RTT$ and the maximum number of information packets we allow to overlap, $\ov$.

Denote the end of the window, EW, by \RTT-1. Hence, the sender transmits $k=\RTT-1$ new information packets (using coded packets $c_t$) before it repeats the same RLNC combination $m=\lfloor p_e \cdot k \rceil$ times\footnote{\label{note1}By the same RLNC combination, we mean that the information packets are the same, but with new random coefficients.}. This is because $p_e \cdot k$ coded packets will be erased on average per $k$ transmitted packets over the channel. Thus, by using FEC of $m$ coded packets in advance, the mean in order delay is reduced. In Fig. \ref{fig:foobar}-(d), which describes an example realization of AC-RLNC, we show EW for each generation using different colors.

Denote by E$\ov$W the end of overlap window, i.e., the number of information packets we allow to overlap. By limiting the maximum value of E$\ov$W and transmitting the same RLNC combinations at E$\ov$W, we can bound the maximum in order delay of the algorithm. The window structure affects the in order delivery delay. We provide analytical results for in order delivery delay in Sec.\ref{bounds}.

Denote by $w \in \{1, \ldots, \ov\}$ the effective window size for the coded combination $c_t$ at time slot $t$. Let $w_{\min}$ denote the actual index of the first information packet in $c_t$, i.e. $w_{\min}-1$ is the index of the last information packet declared as decoded at the sender. The value of $w$ is adaptively determined based on the retransmission criterion, as we define in Sec. IV-f\off{\ref{criterion}}.

%%%
\paragraph{Coded Packet}
This is a causal random linear combination of a subset of information packets\footnote{The information packets are made available to the transport layer (such as a TCP/IP transport layer or an OSI transport layer (Layer 4)) and are not declared by the sender as decoded at the receiver according to the acknowledgments received over the feedback in the previous time slot $t-1$.} within the effective window. The RLNC coded packet at time slot $t$, $c_t$ is given as a function of information packets by
\begin{align}
\label{ct_formula}
c_t = \sum_{i=w_{\min}}^{w_{\min}+w-1} \mu_i \cdot p_i,
\end{align}
where $\mu_i \in \mathbb{F}_z$ are the random coefficients, and $\{p_i\}_{i= w_{\min}}^{w_{\min}+w-1}$ is the subset of information packets within the effective window. We denote by DoF($c_t$) the DoF contained in $c_t$, i.e. the number of distinct information packets in $c_t$. The advantage of using RLNC was given in Sec. \ref{back}.

%%%
\paragraph{Tracking the Channel Rate and the Rate of DoF}
We assume the feedback channel is noiseless\footnote{In the case that the feedback channel is noisy, we can consider a cumulative feedback (for example as given in \cite{malak2019tiny}) that  provides information about the DoF of seen packets, which implies that the previous packets had already been seen if the current packet has been seen.}. Hence, for each $t$-th coded packet transmitted by the sender, the receiver reliably transmits $\ACK(t)$ or $\NACK(t)$ after $\RTT$. Upon the acknowledgements, the sender tracks the actual rate of the channel (as given in Sec. IV-a\off{\ref{ch_behavior}}) and the DoF rate $d=m_d/a_d$, where $m_d$ and $a_d$ denote the number of DoFs needed by the receiver to decode $c_t$, and the number of DoFs added to $c_t$, respectively.

Now we provide one way to calculate $m_d$ in order to manage the delay-throughput tradeoff. Let $\mathcal{A}$ and $\mathcal{N}$ be the sets of ACKs and NACKs within the current window $\mathcal{W}$, respectively.
Moreover, let $\mathcal{F}$ and $FB-\mathcal{F}$, be the sets of slots in which FEC and FB-FEC are sent in the current window $\mathcal{W}$, respectively. In Sec. IV-e\off{\ref{fec}} we will consider the insertion of FEC and FB-FEC slots. Let $\mathcal{P}_t$ denote the set of packets in the coded packet sent in slot $t$. Hence,
    $$m_d = \#\{t : (t \in \mathcal{N} \  \land t \notin FB-\mathcal{F} \  \land t \notin \mathcal{F} ) \  \land \  (\mathcal{W} \cap \mathcal{P}_t) \neq \emptyset\}$$
is the number of NACKed slots in which the coded packets contain new information packets in the current window. These coded packets do not correspond to the FEC or FB-FEC packets.

The sender adds DoFs, $a_d$, by FEC and FB-FEC in the current window $\mathcal{W}$ as we will describe in Sec. IV-e\off{\ref{fec}}. Hence, the number of added DoFs given by,
\begin{align}
a_d = \#\{t : (t \in \mathcal{F} \  \lor \  t \in FB-\mathcal{F}) \  \land \  (\mathcal{W} \cap \mathcal{P}_t) \neq \emptyset\}\nonumber
\end{align}
is the number of slots in which FEC or FB-FEC is transmitted in the current window. In Appendix \ref{App:adaptivecausalRLNC} we provide a detailed description of the example given in Fig. \ref{fig:foobar}-(d) with details on how $m_d$ and $a_d$, and hence $d$ is calculated.

%%%
\paragraph{Insertion of FEC and FB-FEC}\label{fec}
We include two different forward error correction (FEC) mechanisms to add DoFs according to the actual rate of the channel. One is a priori and the second one is posteriori. This on one hand, will provide to the sender sufficient number of DoF to be able to decode the coded packets, minimizing the in order delivery delay. This, on the other hand, will minimize the redundant packets sent by the sender to maximize the throughput.

The first is a priori forward error correction mechanism, denoted by FEC, is sent DoFs in advanced decision according to the average channel rate. Upon the reception of the feedback, if the sender is at the end of the window (i.e. EW), it repeats the same RLNC combinations $m$ times. Note that we determine the number of FECs $m$ adaptively, according to the average erasure rate $e/t$ calculated by the information given over the feedback channel. However, we can adjust the number of $m$ to manage the delay-throughput tradeoff. Increasing $m$ we may reduce the delay. If the sender transmits redundant DoFs (due to the variation in the estimation during the round trip delay) that are not required by the receiver, the throughput will reduce.

The second is posteriori forward error correction mechanism denoted by feedback-FEC (FB-FEC). The insertion of the DoFs is determined by a retransmission criterion defined in Sec. IV-f\off{\ref{criterion}}. Using this FB-FEC mechanism the algorithm guarantees that the receiver obtains sufficient DoFs to be able to decode the coded packets, and maximizes throughput, which however is at the expense of increased in order delay than the duration of the effective window.
\begin{figure*}[t!]
    {\includegraphics[width=1\textwidth]{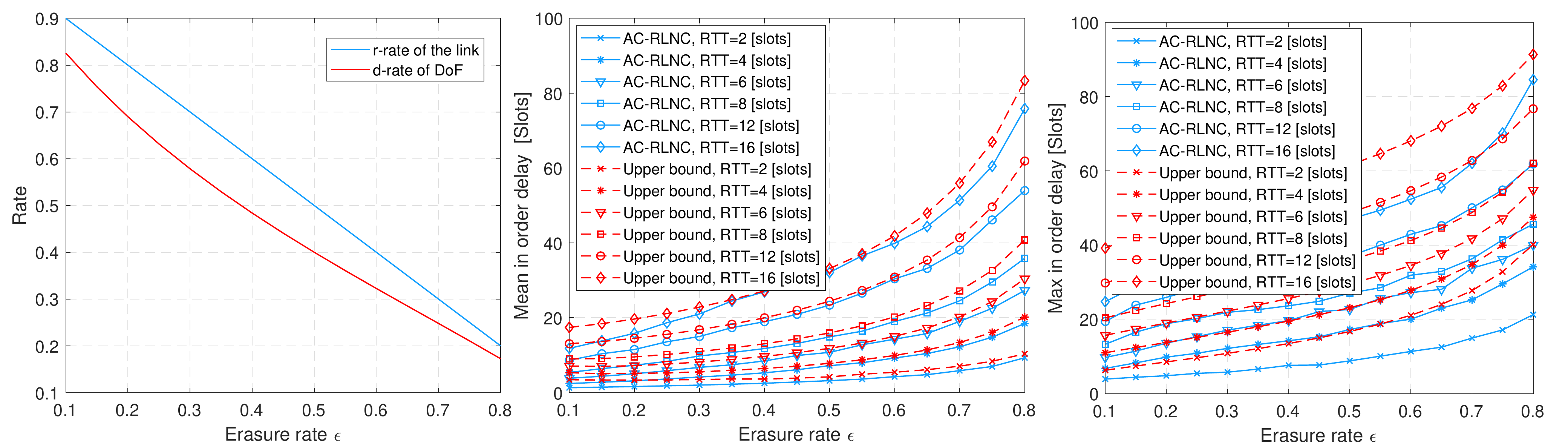}}
	\caption{Delay bounds. Rate (left), mean in order delay (middle), and maximum in order delay (right) for $o=2k$ and $P_e=10^{-3}$. The bounds (red dashed lines) presented have a good agreement with the AC-RLNC simulation results (blue solid lines) where $th=0$. \mm{Note that an RTT of two time slots is the theoretical minimum that can be achieved in a point-to-point communication system with feedback.}}
	\label{fig:DB}
\end{figure*}

%%%
\paragraph{Transmission Criterion}\label{criterion}
The retransmission criterion depends on whether the feedback is an ACK or a NACK. If the channel rate $r$ is sufficiently higher than the required DoF rate $d$ (which is given by the ratio of the number of DoFs needed to decode $c_t$ and the number of DoFs added to $c_t$), then the retransmission condition is satisfied:
\[
    r-d>th,
\]
where $th$ is the threshold. In this case, a new packet $p_i$ is added to the random linear combination if it is not the EW, and otherwise the same random linear combination is transmitted.

By setting $th=0$, the retransmission criterion becomes $r>d$. In this case, the algorithm will track the average rate of the channel. However, we can set the threshold $th$ adaptively according to the maximum in order delivery delay requirements of the applications, and the standard deviation of the erasure events. For example, in order to support the burst of erasures and lower the maximum in order delivery delay, we can compute the second moment of the erasures (and denote by $v_e$) such that the threshold is set to be $th= \sqrt{v_e}$. In general, we can choose the threshold adaptively such that $th \leq \sqrt{v_e}$ to manage the throughput-delay tradeoffs. Moreover, one can use the burst pattern to adapt the transmission criteria, which is left as future work.

%%%%%%%%%%%%%%%%%%%%%%%%%%%%%%%%%%%%%%%%%%%%%%%%%%%%%%%%%%%%%%%
\section{Analytical Results for Delay and Throughput}\label{bounds}
%%%%%%%%%%%%%%%%%%%%%%%%%%%%%%%%%%%%%%%%%%%%%%%%%%%%%%%%%%%%%%%
In this section we derive upper bounds for the mean and maximum in order delivery delay and throughput for AC-RLNC proposed in Sect. \ref{algo}. First, note that in the code suggested herein we bound the number of distinct information packets in $c_t$ by $\ov$.
Limiting the maximum number of information packets that can overlap we reduce the mean and limit the maximum in order delay. Hence, in the analytical results concerning the in order delay, we consider the maximum lengths of the effective window and the end window of $\ov$ new packets, denoted by $w_{\max}$ and E$\ov$W, respectively.

%%%%%%%%%%%%%%%%%%%%%%%%%%%%%%%%%%%%%%%%%%%%%%%%%%%%%%%%%%%%%%%
\subsection{An Upper Bound to Mean In Order Delivery Delay}
%%%%%%%%%%%%%%%%%%%%%%%%%%%%%%%%%%%%%%%%%%%%%%%%%%%%%%%%%%%%%%%
Following the notation in Table \ref{fig:table}, and Algorithm \ref{causalRLNCalgo}, the rate of DoF is given by
\begin{align*}
d = \frac{m_d}{a_d},
\end{align*}
 where the number of DoFs needed by the receiver to decode $c_t$, i.e. $m_d$, satisfies
\begin{align*}
m_d = \ov\epsilon,
\end{align*}
and the DoF added $a_d$ to $c_t$ satisfies
\begin{align*}
a_d = \frac{1}{1-\epsilon} m_d + \epsilon m_e = \frac{1}{1-\epsilon} \ov\epsilon + \epsilon \ov\epsilon,
\end{align*}
where $m_e = \ov\epsilon$ is the effective number of the DoFs required by the receiver, and $k =  \RTT-1$ is defined in Table \ref{fig:table}.

The condition for retransmission is $r>d+th$, such that
\begin{eqnarray*}
1-r &=&  \epsilon \\
    &<& 1-d-th.
\end{eqnarray*}
Hence, we get that $\epsilon<\epsilon_{\max}\leq 1-d-th$, where $\epsilon_{\max}$ is an upper bound to the erasure probability of the forward channel.

To be able determine the mean and maximum in order delivery delays, $\Dm$ and $\DM$, respectively, we next need to calculate the following probabilities:

\noindent (1) {\bf Condition for starting a new generation.} Probability that it is E$\ov$W is:
\begin{equation}
\label{cond_EoW}
    \mathbb{P}_{E\ov W} = (1-\epsilon_{\max})^{\ov}.
\end{equation}

\noindent (2) {\bf Condition for retransmission.}  Probability that $r>d+th$ over two windows is:
\begin{equation*}
   \mathbb{P}_{r>d+th} = \sum_{i=1}^{\lfloor \ov\epsilon_{\max}\rfloor} \binom{\ov}{i}\epsilon^{i}(1-\epsilon)^{\ov -i}.
\end{equation*}
Finally, we calculate upper bounds for the mean in order delay for BEC under different feedback states.

%%%
\paragraph{No Feedback}
Given that there is no feedback, we have
\begin{multline}
\label{no_FB_BEC}
 {\Dm}_{[no\mbox{ }feedback]} \leq \frac{1}{1-\epsilon_{\max}}\\ \Big[ \mathbb{P}_{E\ov W}(m_e+k)+(1-\mathbb{P}_{E\ov W})\RTT \Big],
\end{multline}
in which if it is E$\ov$W (i.e. (\ref{cond_EoW}) is satisfied), the same RLNC is transmitted $m_e$ times, yielding a delay of $m_e+k$, and if it is not E$\ov$W, a new $p_i$ is added to the RLNC and transmitted, yielding a delay of $\RTT=k+1$. The scaling term in the upper bound $\frac{1}{1-\epsilon_{\max}}$ is due to the maximum number of retransmissions needed to succeed in the forward channel. These steps follow from Algorithm \ref{causalRLNCalgo} by replacing EW with E$\ov$W.

%%%
\paragraph{NACK}
When the feedback message is a NACK, we have
\begin{multline}
\label{NACK_BEC}
{\Dm}_{[nack\mbox{ }feedback]} \leq \epsilon_{\max}\frac{1}{1-\epsilon_{\max}}\\ \Big[ \mathbb{P}_{r>d+th}\big[ (1-\mathbb{P}_{E \ov W})\RTT+\mathbb{P}_{E\ov W}(m_e+k)\big] \\  +(1-\mathbb{P}_{r>d+th})\big[ \RTT + \mathbb{P}_{E\ov W}(m_e+k) \big]\Big],
\end{multline}
which follows from that given $r>d+th$, which is with probability $\mathbb{P}_{r>d+th}$, the mean in order delay is same as the case when there is no feedback, and if $r\leq d + th$, which is with probability $(1-\mathbb{P}_{r>d+th})$, the same RLNC is transmitted, and then if it is E$\ov$W, the same combination is retransmitted $m_e$ times. Here, the scaling term $\frac{1}{1-\epsilon_{\max}}$, similarly to the no feedback case, is due to the maximum number of retransmissions needed to succeed in the forward channel.

%%%
\paragraph{ACK}
When the feedback message is an ACK, we obtain
\begin{multline}
\label{ACK_BEC}
{\Dm}_{[ack\mbox{ }feedback]} \leq (1-\epsilon_{\max})\Big[ \mathbb{P}_{E\ov W}(m_e+k)\\  + (\mathbb{P}_{r>d+th})\RTT  + (1-\mathbb{P}_{r>d+th})\RTT\Big],
\end{multline}
which is due to that when it is the E$\ov$W, the same RLNC is transmitted $m_e$ times. Then, if $r\geq d + th$, which is with probability $\mathbb{P}_{r>d+th}$, then a new packet is added to the RLNC and transmitted. Otherwise, the same RLNC is transmitted, and both cases yield a delay of $\RTT=k+1$. Since the feedback is an ACK, the mean in order delay we compute is scaled by $1-\epsilon_{\max}$, which is a lower bound on the probability of getting an ACK with perfect feedback.

Given the round trip delay, we do not have feedback in the first transmission window. Hence, to normalize the effect of not having feedback, we use a normalization parameter $\lambda $ denoting the fraction of time there is feedback such that
\begin{multline*}
\Dm \leq \lambda {\Dm}_{[no\mbox{ }feedback]}\\
+(1-\lambda )({\Dm}_{[nack\mbox{ }feedback]}+{\Dm}_{[ack\mbox{ }feedback]}).
\end{multline*}

We use this upper bound in Sect. \ref{simulation} to validate our numerical simulations where $\ov=2k$.

%%%%
For the GE channel model, erasure events only occur when the forward channel is in state $B$. Therefore, the average number of transmissions in the forward channel can be computed using the following relation
\begin{align}
\label{GE_average_transmissions}
\pi_G+\sum\limits_{k=2}^{\infty} \pi_B(1-s)^{k-2} s k
\end{align}
where the first term denotes the fraction of time the channel is state $G$, for which only one transmission is required, i.e. $k=1$, and the term inside the summation  denotes the probability that the channel starts in state $B$ and transits to state $G$ in $k\geq 2$ slots. Evaluating (\ref{GE_average_transmissions}), along with $\pi_{B}=1-\pi_{G}=\epsilon$ the number of retransmissions needed to succeed in the forward channel is
\begin{align}
1+\epsilon \left[\left(\frac{1}{1-s}\right)\left(\frac{1}{s}-s\right)-1\right],\quad s\in(0,1].
\end{align}
%where $s=q(\epsilon^{-1}-1)$.
We can show that if $\frac{1}{s}-s-1 >\frac{1}{1-\epsilon}$, the number of retransmissions needed for GE channel is higher than the number of retransmissions for BEC. Since we mainly focus on the bursty GE channel model, i.e. when $s$ is small such that
\begin{equation}
    \left(\frac{1}{1-s}\right)\left(\frac{1}{s}-s\right)-1 >\frac{1}{1-\epsilon},
\end{equation}
we can show that the number of retransmissions required in this case will be higher compared to the BEC case. Using this consideration and the similar upper bounding techniques as in the case of BEC (see (\ref{no_FB_BEC}), (\ref{NACK_BEC}) and (\ref{ACK_BEC})), we can show that the upper bound for the mean in order delay for the GE channel is higher.
\begin{figure}
    %{\includegraphics[trim=10cm 8.5cm 1cm 6.55cm,clip,scale=0.62]{figures/max_bound2.jpg}}
    {\includegraphics[width =\columnwidth]{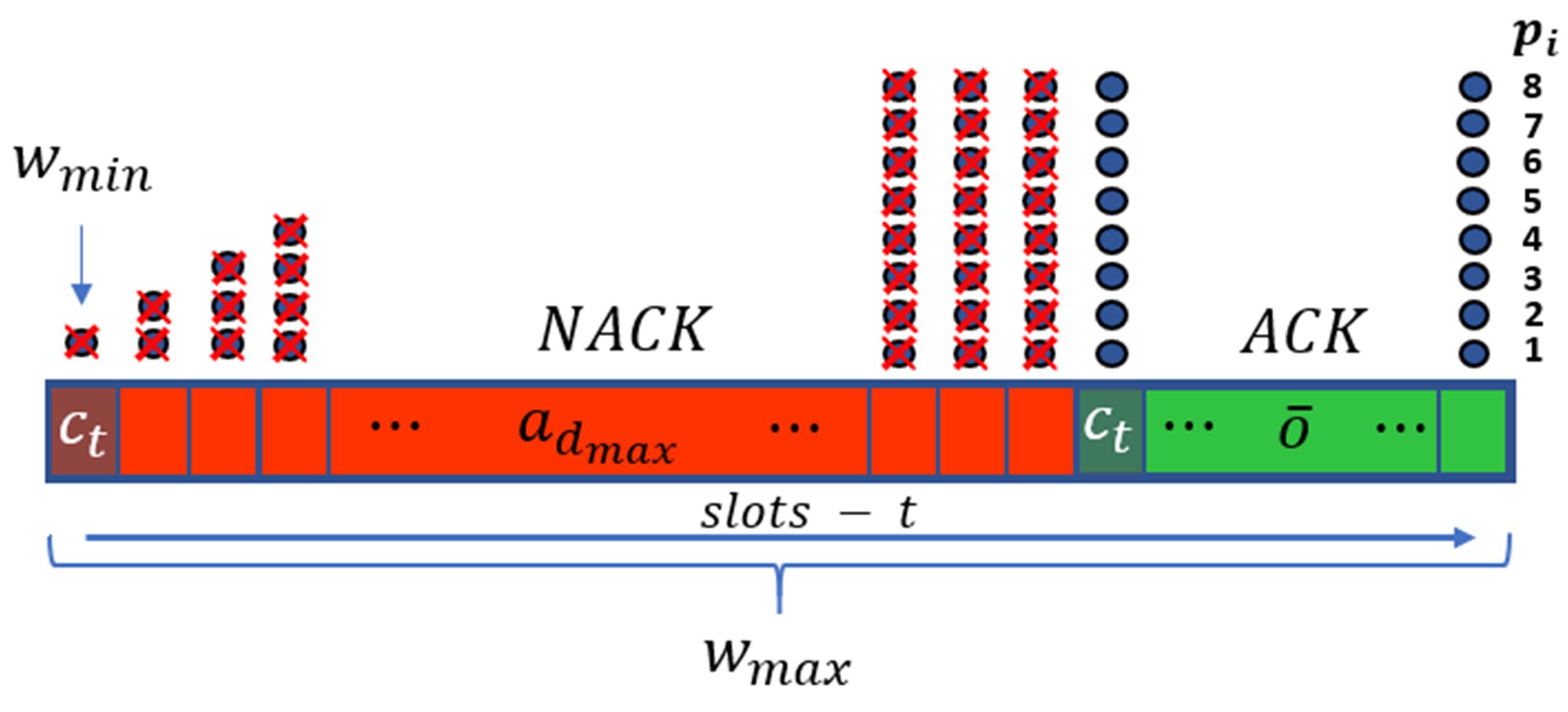}}
	\caption{Maximum in order delay bound. The red interval represents the slots with forward channel erasures, and the green one represents the slots that are not erased. In the left corner, $w_{\min}$ points out the slot of the first information packet in $c_t$ that is first transmitted. Since the sender transmits coded packets, once there are $\ov$ information packets in the coded packet, the order of erasures in the forward channel do not affect the in order delay. In this case the receiver will need $\ov$ DoFs to be able to decode a coded packet. Hence, the worst case is the bursty one in which right after the transmission of the first packet, the sender will add $\ov-1$ information packets to the coded packets that will not be delivered due to the consecutive erasures in the forward channel. In this case, all the $a_{d_{\max}}$ coded packets (which include at most $\ov$ information packets) are erased in the forward channel. Then, to decode these information packets, the receiver needs to obtain $\ov$ coded packets (green interval on the right). Hence, we define $w_{\max}$ interval length as $a_{d_{\max}}+\ov$.}
	\label{fig:max_delay}
\end{figure}

%%%%%%%%%%%%%%%%%%%%%%%%%%%%%%%%%%%%%%%%%%%%%%%%%%%%%%%%%%%%%%%
\subsection{An Upper Bound to Maximum In Order Delivery Delay}
%%%%%%%%%%%%%%%%%%%%%%%%%%%%%%%%%%%%%%%%%%%%%%%%%%%%%%%%%%%%%%%
In the algorithm proposed, the maximum number of information packets in $c_t$ is limited. Thus, when DoF($c_t$) $=\ov$, the sender transmits the same RLNC combination until all $\ov$ information packets are decoded. In this case, since each transmitted packet is a coded combination, any $\ov$ packets delivered at the receiver are sufficient to decode $c_t$. Let $w_{\max}$ denote the time interval between when the first information packet in $c_t$ is first transmitted and when all the $\ov$ information packets in $c_t$ are decoded at the receiver. This time interval may also include at most $a_{d_{\max}}$ coded packets that are erased in the forward channel, in addition to $\ov$ coded packets successfully delivered to the receiver. Hence, $w_{\max}= \ov + a_{d_{\max}}$ \mm{as presented in Fig. \ref{fig:max_delay}}. The maximum delay that the first information packet can experience is when there are $a_{d_{\max}}$ erasures first, and then the $\ov$ are successfully delivered next (i.e., when the channel is bursty).

Denote by $P_e$ the probability of error which is when there are more than $a_{d_{\max}}$ packets that are erased in $w_{\max}$. Hence,
\[
    \mathbb{P}_e \leq \epsilon^{a_{d_{\max}}} = \epsilon^{w_{\max}-\ov}.
\]
Rearranging terms results in
\[
    w_{\max} \geq \log_{\epsilon_{\max}}(\mathbb{P}_e) + \ov.
\]
Now, since the maximum number of missing DoFs in $w_{\max}$ is $\ov\epsilon_{\max}$, the maximum in order delay is bounded by
\[
    \DM \leq \ov\epsilon_{\max} + \log_{\epsilon_{\max}}(\mathbb{P}_e) + \ov
\]
for any selected error probability $\mathbb{P}_e$.

In Fig. \ref{fig:DB} we present the results for the rates of the DoF $d$ and the channel $r$ calculated at the sender, and the mean and the maximum in order delivery delays bounds $\Dm$ and $\DM$ for $\ov=2k$ and $P_e=10^{-3}$. The analytical results presented have a good agreement with the simulation results given in Sect. \ref{simulation}, where $th=0$.

%%%%%%%%%%%%%%%%%%%%%%%%%%%%%%%%%%%%%%%%%%%%%%%%%%%%%%%%%%%%%%%
\subsection{An Upper Bound to Throughput}
%%%%%%%%%%%%%%%%%%%%%%%%%%%%%%%%%%%%%%%%%%%%%%%%%%%%%%%%%%%%%%%
In the adaptive code proposed in Section \ref{algo}, the sender learns the rate of the channel and the rate of the DoF according to the acknowledgements obtained over the feedback channel. However, it is essential to note that due to the RTT delay, those rates at the sender are updated with delay. Hence, at time slot $t$, the actual retransmission criterion at the sender is calculated as
\begin{align}
\label{retransmission_critetion}
 r(t^-)-d(t^-)>th(t^-),
\end{align}
where $t^-=t-\RTT$. \mm{Hypothetically, if} $\RTT$ is less than 1 slot (i.e. $\RTT<1$), the adaptive code is able to obtain the rate of the channel.  \mm{However, in the non-asymptotic model considered, RTT delay is higher, i.e. $\RTT\geq 2$. Hence, we may} get degradation on the throughput due to the variations in the channel. This is because those variations are not reflected in the retransmission criterion at the sender in time slot $t$. By bounding the channel variance during RTT, we can provide bounds on the throughput. Different bounds for the variations of the channels are considered in the literature \cite{polyanskiy2010channel,polyanskiy2011dispersion}. Bounding the variance, we will obtain the maximum variation between the channel rate calculated at the sender for the retransmission criterion in (\ref{retransmission_critetion}) to the actual channel rate.
\begin{figure}
    {\includegraphics[width =\columnwidth]{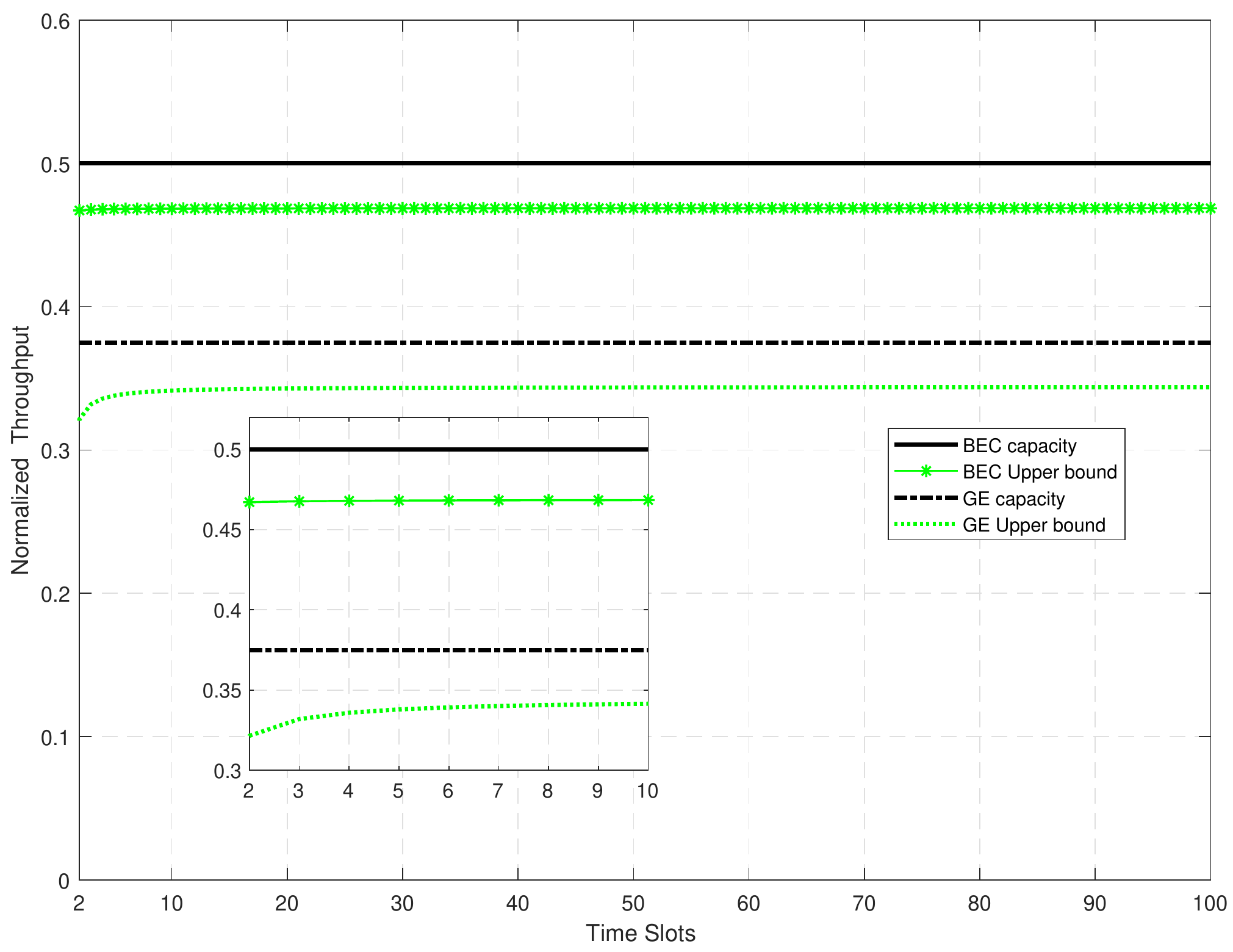}}
	\caption{Throughput upper bounds for BEC and GE channels with erasure probability of 0.5 ($\epsilon=0.5$). \mm{Note that the range of the abscissa is from $2$ (the theoretical minimum of the $\RTT$ delay) to $100$. Moreover, the throughput is not degraded by increasing $\RTT$. In the asymptotic regime the AC-RLNC code may attain the capacity.}}
	\label{fig:ThroughputError05BecGe}
\end{figure}

In the code suggested, the calculated rate sets the number of RLNC coded packets with the same or new information packets to be transmitted during the period of RTT. We denote by $\textbf{c}=(c_{t^-},\ldots,c_{t})$ the vector of the RLNC packets transmitted during a period of RTT transmissions given the estimated rate $r(t^-)$. Denote by $\textbf{c}^{\prime}=(c^{\prime}_{t^-},\ldots,c^{\prime}_{t})$ the vector of the RLNC packets transmitted if the actual rate of the channel $r(t)$ was available at the sender non-causally.
\begin{figure}
    {\includegraphics[width =\columnwidth]{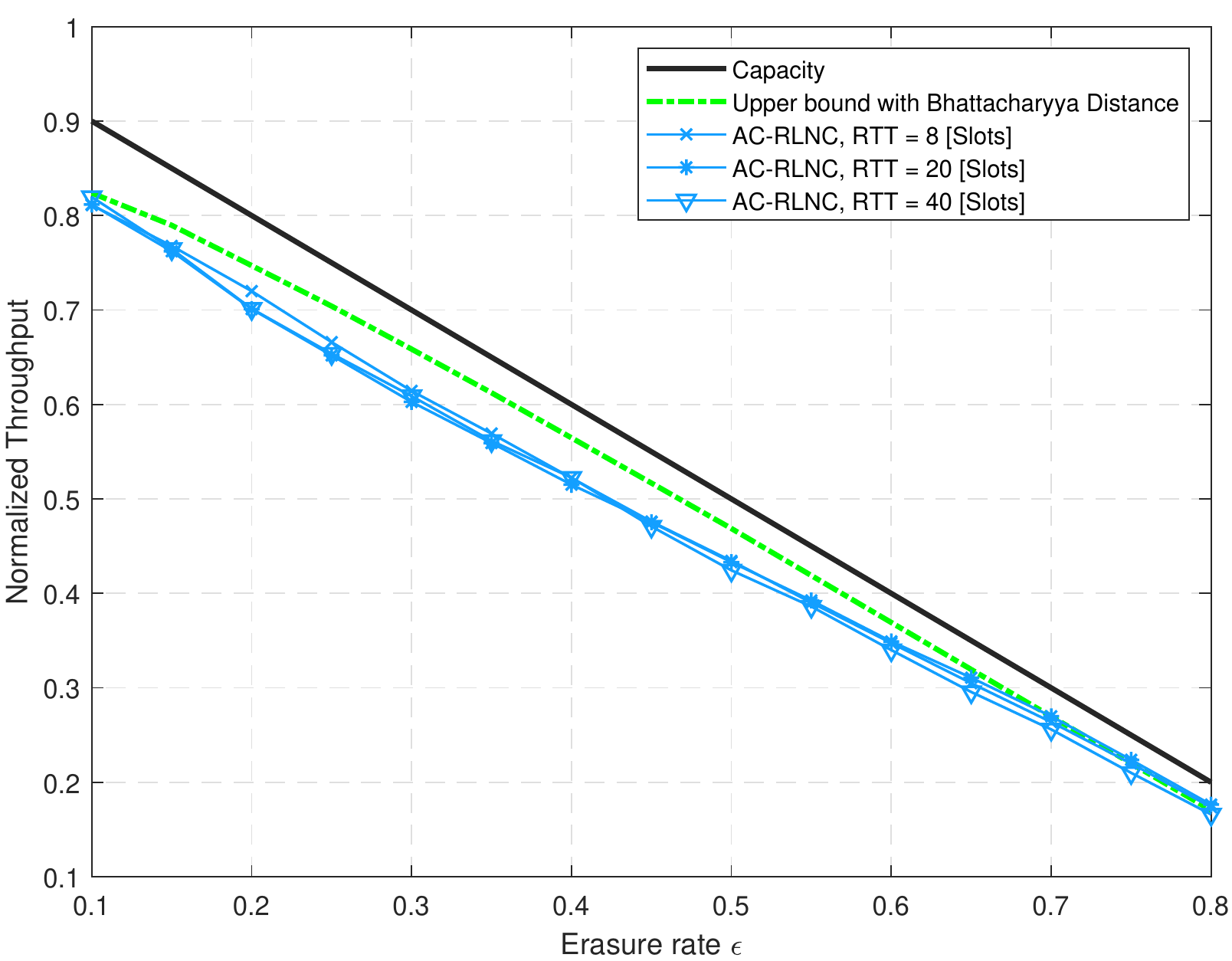}}
	\caption{Throughput upper bound for BEC.}
	\label{fig:Throughput}
\end{figure}
\begin{figure}
    {\includegraphics[width =\columnwidth]{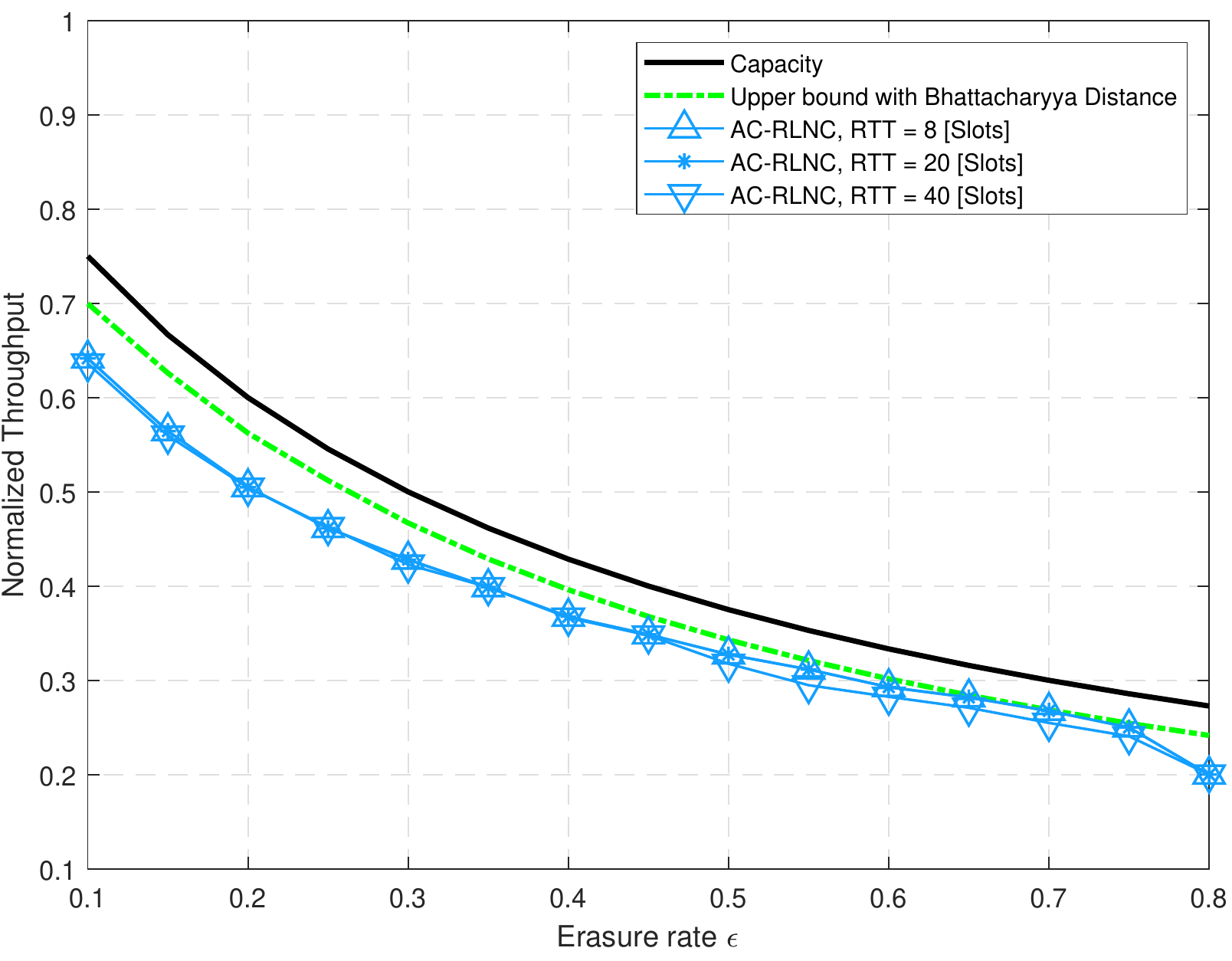}}
	\caption{Throughput upper bound for GE channel.}
	\label{fig:Throughput_GE}
\end{figure}

Now we consider the case where the actual rate of the channel, $r(t)$, at time slot $t$ is higher than the rate, $r(t^-)$, estimated at the sender at time slot $t^-$. In this case, the sender will transmit additional DoF (additional RLNC coded packets of the same information packets with different coefficients) which are not required at the receiver to decode. In the case that the estimated rate is lower than the actual rate of the channel, we don't lose throughput because the sender does not transmit redundant DoFs. However, since there are missing DoFs at the receiver to decode, the in order delay increases. The number of the additional DoF not required by the receiver, during RTT time slots, will be determined according to the distance between the number of packets not erased at the receiver given $r(t)$, to the estimated number of non-erased packets given $r(t^-)$.
\begin{figure*}[t!]
	\centering
	{\includegraphics[width=\textwidth]{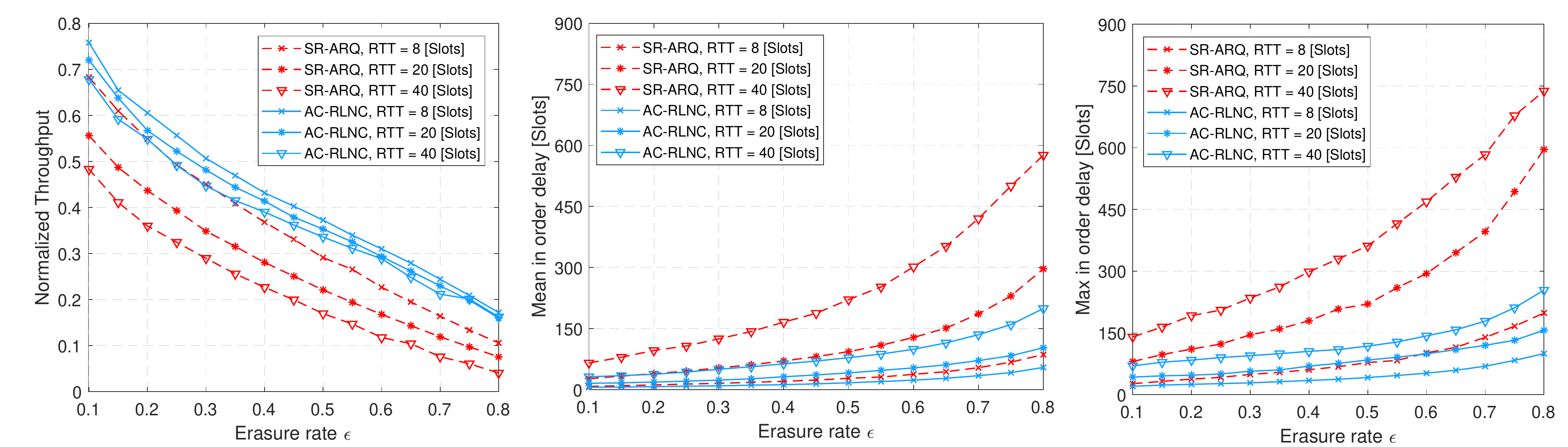}}
	\caption{AC-RLNC simulation for BEC (memoryless). Throughput (left), mean in order delay (middle), and maximum in order delay (right).}
	\label{fig:BEC}	
	{\includegraphics[width=\textwidth]{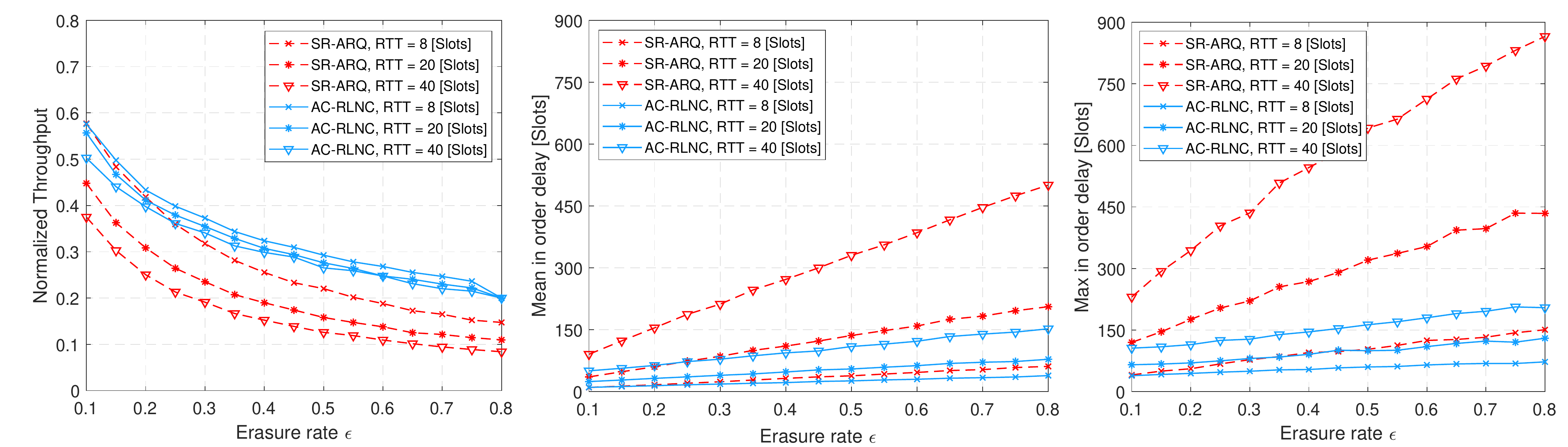}}
	\caption{AC-RLNC simulation for GE channel. Throughput (left), mean in order delay (middle), and maximum in order delay (right).}
	\label{fig:GE}	
\end{figure*}

We next define the Bhattacharyya distance which we use to provide an upper bound to the throughput. In \cite{shannon1967lower,viterbi2013principles,dalai2014elias,barg2005distance}, Bhattacharyya Distance was considered to bound the zero error capacity.
\begin{defi}
The Bhattacharyya distance is given by \cite{barg2005distance}
\begin{equation}\label{eq:BD}
    l(c,c^{\prime})=-ln(BC(c,c')),
\end{equation}
where $BC(c,c')$ is the Bhattacharyya coefficient, given as
\begin{align}\label{eq:BC}
  BC(c,c') = \sum_{y}\sqrt{W(y|c)W(y|c^{\prime})},
\end{align}
and, $W(y|c)$ and $W(y|c^{\prime})$ is the channel transition probabilities from inputs $c$ and $c^{\prime}$ to output vector $y$, respectively. We point out that bounds on the Bhattacharyya distance for codes can be immediately mapped to bounds on the reliability function for certain channels.
\end{defi}
\off{Thus Bhattacharyya distance is,
\begin{equation}\label{eq:BD}
    l(c,c') = -ln \Big(BC(c,c')\Big).
\end{equation}}

\begin{theo}{\bf An upper bound on the throughput of AC-RLNC is}
\begin{align}
    \eta\leq r(t^-) - l(r(t),r(t^-)),
\end{align}
where $l(\cdot,\cdot)$ is the Bhattacharyya distance.
\end{theo}

\begin{proof}
One reasonable approach to bound the distance between the number of erasures at the output vector, $y$, of the channel during the period of RTT is using the Bhattacharyya distance between the channel probability distributions at time $t$ and $t^-$.

Given the actual channel rate $r(t^-)$ at time slot $t^-$, the rate of the channel at time slot $t$ is bounded by
\begin{align}
\label{r_upper_bound}
    r(t)\leq r(t^-) + \frac{\sqrt{V(t)}}{k+m},
\end{align}
where $V(t)$ is the variance of the channel during the period of \RTT.

We consider an $\RTT$ period. Therefore, let $W(y\vert c')=r(t^-)$ and $W(y\vert c)=r(t)$, where we use the upper bound for $r(t)$ given in (\ref{r_upper_bound}) instead of computing it explicitly. Letting the summation range in (\ref{eq:BC}) be from $t=0$ to $\RTT-1$, we can upper bound the throughput as
\off{Since we consider an $\RTT$ period, using (\ref{eq:BD}), and letting $W(y\vert c)=r(t)$ (where we use the upper bound for $r(t)$ in (\ref{r_upper_bound}) instead of computing it explicitly), and $W(y\vert c')=r(t^-)$, and letting the summation range in (\ref{eq:BC}) be from $t=0$ to $\RTT-1$, we can upper bound the throughput as}
\begin{align}
    \eta \leq r(t^-) - l(r(t),r(t^-)).\nonumber
\end{align}
\end{proof}

\begin{cor}{\bf BEC.}
An upper bound on the throughput of AC-RLNC for BEC is
\begin{eqnarray*}
    \eta_{BEC}\leq  1-\epsilon - l(r_{BEC}(t),r_{BEC}(t^-)).
\end{eqnarray*}
\end{cor}

\begin{proof}
For the BEC channel, when $r(t^-)=1-p_e(t^-)=1-\epsilon$, we have,
\[
    r_{BEC}(t)\leq r_{BEC}(t^-) + \frac{\sqrt{V_{BEC}(t)}}{k+m},
\]
where $V_{BEC}(t)$ is the variance of the BEC during  the  period  of \RTT. Since the total number of successes in a sequence of $\RTT$ independent realizations of the channel, has a Binomial distribution with mean $\RTT r_{BEC}(t^-)$, $V_{BEC}(t)$ given as
\[
V_{BEC}(t) = \RTT(1-r_{BEC}(t^-))r_{BEC}(t^-).
\]

In BEC $r(t)$ and $r(t^-)$, respectively, be two Binomial distributions with the same parameter $\RTT$. Then, the Bhattacharyya coefficient is given by
\begin{align}
    &BC_{BEC}(r(t),r(t^-))\nonumber\\
    &=\sum_{t=0}^{\RTT-1} \binom{\RTT}{t}\big(r(t)r(t^-)\big)^{\frac{t}{2}}\big((1-r(t))(1-r(t^-))\big)^{\frac{\RTT-t}{2}},\nonumber
\end{align}
and the Bhattacharyya distance for BEC is given by
\begin{equation*}\label{eq:BD_BEC}
    l_{BEC}(r(t),r(t^-)) = -ln \Big(BC_{BEC}(r(t),r(t^-))\Big).
\end{equation*}
Hence the upper bound on the throughput is given by
\begin{align}
    \eta_{BEC}&\leq r(t^-) - l_{BEC}(r_{BEC}(t),r_{BEC}(t^-))\nonumber\\
    &= 1-\epsilon - l_{BEC}(r_{BEC}(t),r_{BEC}(t^-)).\nonumber
\end{align}
%using Bhattacharyya distance given in \eqref{eq:BD}.
\end{proof}

In Fig. \ref{fig:ThroughputError05BecGe}, we show the throughput upper bound of AC-RLNC as function of $\RTT$ for BEC with $\epsilon=0.5$. We can note that in this case the suggested solution may achieve around $93\%$ of the BEC channel \mm{capacity\off{\footnotemark[\ref{note2}]}.} In Fig. \ref{fig:Throughput}, we present the throughput upper bound for BEC as function of $\epsilon$. Note that in the simulation results presented in Fig. \ref{fig:Throughput}, by changing the parameters in the algorithm, we may obtain the upper bound. For example, to do that we can limit the maximum number of information packets we allow in the RLNC to $\ov=4k$. This will increase the in order delay; however, by the parameters of the algorithm suggested we can improve the throughput and manage the throughput-delay tradeoffs according to the constraint of each application.

\begin{cor}{\bf GE Channel.}
An upper bound on the throughput of AC-RLNC for GE channel \ton{where $\epsilon_G=0$} is
\begin{eqnarray*}
    \eta_{GE}\leq 1-\pi_B - l(r_{GE}(t),r_{GE}(t^-)).
\end{eqnarray*}
\end{cor}

\begin{proof}
Let $E$ be the random variable denoting the erasure event. For the GE channel considered the average erasure rate is given by
\begin{equation*}
     \mathbb{E}[E]=\epsilon_G \pi_G + \epsilon_B \pi_B = \pi_B,
\end{equation*}
where $\epsilon_G=0$ and $\epsilon_B=1$. Hence, we assume $r(t^-)=1-p_e(t^-)=\frac{s}{q+s}$.  We denote by $X$ the state random variable.
The variance of the GE channel given by
\begin{eqnarray*}
  V_{GE}(t) &=&  \mathbb{E}(V_{GE}(t)(E|X))\RTT\\
    &=& \mathbb{E}((E-\mathbb{E}(E|X))^2|X)\RTT \\
    &=&  \left( \mathbb{E}[E^2|X]-(\pi_B \epsilon_B)^2 \right)\RTT\\
    &=& \left((\pi_B \epsilon_B)-(\pi_B \epsilon_B)^2\right)\RTT.
\end{eqnarray*}

Now, in the same way that given above for the BEC channel, the rate at time slot $t$ for the GE channel is bounded by
\[
    r_{GE}(t)\leq r_{GE}(t^-) + \frac{\sqrt{V_{GE}(t)}}{k+m},
\]
and the throughput by is bounded by
\begin{eqnarray*}
    \eta_{GE}&\leq& r(t^-) - l(r_{GE}(t),r_{GE}(t^-))\\
    &=& 1-\pi_B - l(r_{GE}(t),r_{GE}(t^-)),
\end{eqnarray*}
using Bhattacharyya distance given in \eqref{eq:BD}.
\begin{figure*}
	\centering
	{\includegraphics[width=0.78\textwidth]{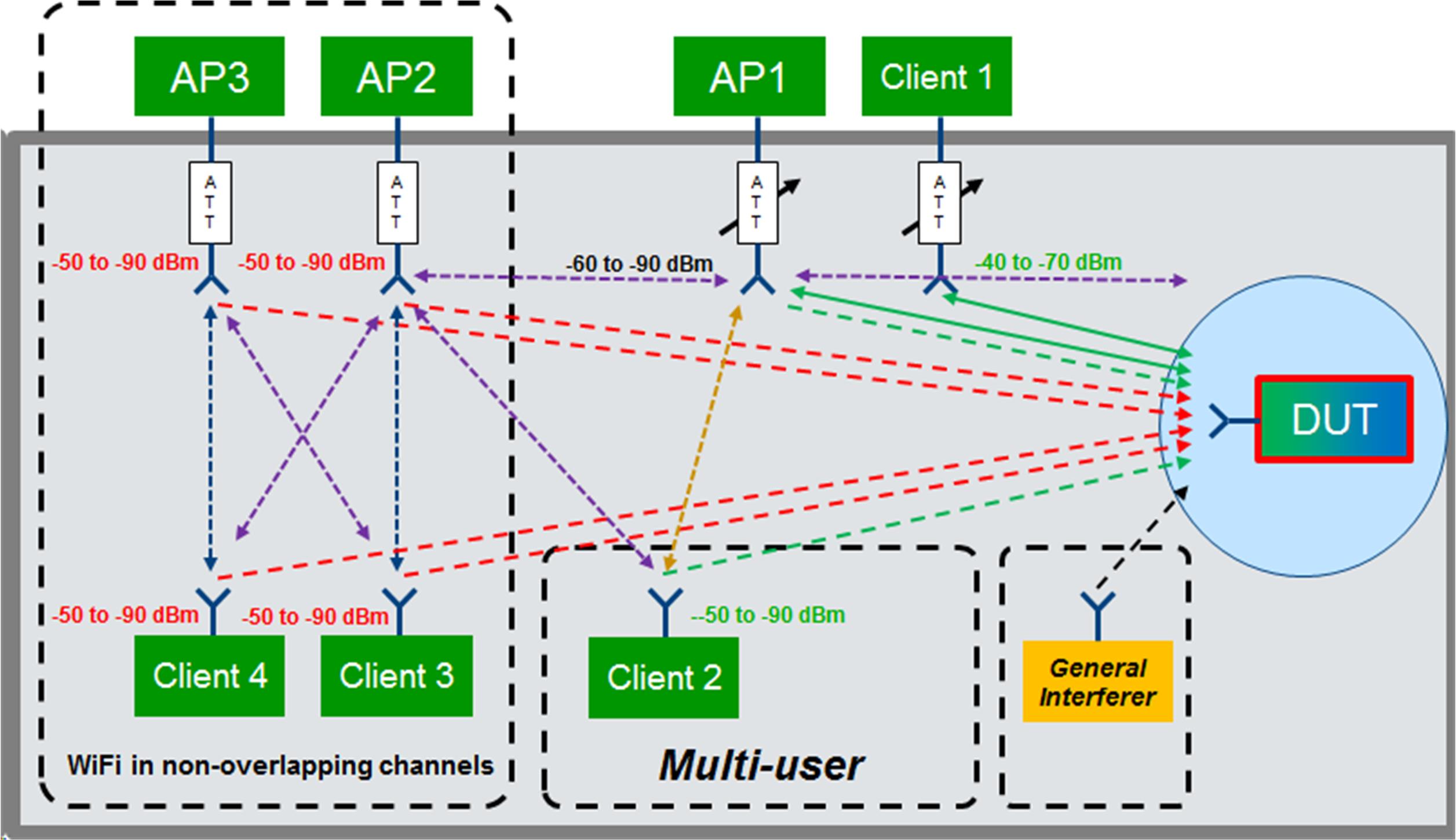}}
	\caption{Controlled-congested setup considered by Intel. In this setup, there are three access points (APs), four clients (one of them is multi-user), where the Device Under Test (DUT) is the receiver, and one general interferer (no WiFi or Bluetooth (BT)) using a signal generator. The dashed lines to the DUT are the point-to-point communication channels analyzed. The WiFi standards used are IEEE 802.11n, IEEE 802.11ax and IEEE 802.11ac. The transmit powers can be adjusted to test several possible scenarios and environments. The traces were collected from all senders in different channel conditions (i.e. controlled by the channel power gains). Those traces include the packets transmitted in each time slot and the acknowledgments obtained over the feedback channel from the receiver.}
	\label{fig:Intel_Setup}
\end{figure*}

%%ADD
We can upper bound the Bhattacharyya distance of the GE channel using the Bhattacharyya distance of the BEC as
\begin{align}
 l(r_{GE}(t),r_{GE}(t^-)) \leq l_{BEC}(r(t),r(t^-)).
\end{align}
The above relation is due to that since we have $\pi_G=1-\epsilon=s/(s+q)$. This implies that $1-q$ is large compared to $q$, and when the channel state is $G$, the probability that the channel stays there will be higher. As a result, the number of successes in a sequence of $\RTT$ dependent realizations of the channel is positively correlated, and the probability distribution of the number of successes satisfies
\begin{align}
    r_{GE}(t)\geq r_{BEC}(t),
\end{align}
yielding
\[
BC_{GE}(r(t),r(t^-))\geq BC_{BEC}(r(t),r(t^-)).
\]
\end{proof}

In Fig. \ref{fig:ThroughputError05BecGe}, we show the throughput upper bound of AC-RLNC as function of $\RTT$ for GE channel with \ton{$q = \epsilon$} and $s=0.3$. We can note that in this case the suggested solution may achieves around $91\%$ of the channel \mm{capacity\off{\footnotemark[\ref{note2}]}}. In Fig. \ref{fig:Throughput_GE}, we present the throughput upper bound for the GE channel where \ton{$q = \epsilon$ and} $s=0.3$. Similar to BEC results presented in Fig. \ref{fig:Throughput}, we limit the maximum number of information packets we allow in the RLNC to $\ov=4k$. This is to show that by changing the parameters of AC-RLNC we may obtain the upper bound.
\off{
\ton{In the second case, we will consider, we assume that in the first period of RTT in a window the rate $r(t^-)$ estimated at the sender in time index $t-\RTT$ is slightly lower that the actual rate $t(t)$ at the channel in the time index $t$.
Now, in the second period of the RTT we assume the was higher as considered in the first case.
}
%%%%%%%%%%%%%%%%%%%%%%%%%%%%%%%%%%%%%%%%%%%%%%%%%%%%%%%%%%%%%%%
\subsection{Adaptive Threshold}
%%%%%%%%%%%%%%%%%%%%%%%%%%%%%%%%%%%%%%%%%%%%%%%%%%%%%%%%%%%%%%%
First, in each window $i$ according the $\RTT$ we will calculate the variance, $v_e(i)$, of the channel. For example, for each window, in the BEC since the erasure probability is Bernoulli i.i.d.,
\[
v_e(i) = \frac{1}{\RTT}\sum_{l=1}^{\RTT} (\textbf{1}_{c_l}- e/t)^2
\]
where $\textbf{1}_{c_l}$ is the indicator for the event that the coded packet $c_l$ is not delivered to the receiver (namely erased).

Then we will calculate the recursive average of the variance by the forgetting factor $0 \leq \lambda \leq 1$, such that
\[
\Tilde{v}_e(i) = \lambda \Tilde{v}_e(i-1) + (1-\lambda) v_e(i).
\]
To conclude we will choose the threshold $th$, such that, $th \leq \Tilde{v}_e(i)$ and the forgetting factor $\lambda$ according the delay constraint of each application. When $\lambda$ is close to $1$ we will support the average erasure probability of the channel. As we reduce the value of $\lambda$ we support the fluctuations of the channel.}

%%%%%%%%%%%%%%%%%%%%%%%%%%%%%%%%%%%%%%%%%%%%%%%%%%%%%%%%%%%%%%%
\section{Experimental Study and Simulation Results}\label{simulation}
%%%%%%%%%%%%%%%%%%%%%%%%%%%%%%%%%%%%%%%%%%%%%%%%%%%%%%%%%%%%%%%
In this section, we demonstrate the performance of AC-RLNC with feedback. We first simulate the selective repeat ARQ (SR-ARQ) protocol, and validate the throughput $\eta$, the mean in order delivery delay $\Dm$, and the maximum in order delivery delay $\DM$ performance. We then validate the simulation results via experimental simulation results using the Intel traces for a point-to-point wireless communication system. The controlled-congested setup considered by Intel is illustrated in Fig. \ref{fig:Intel_Setup} along with its description.
\begin{figure*}
    {\includegraphics[width=\textwidth]{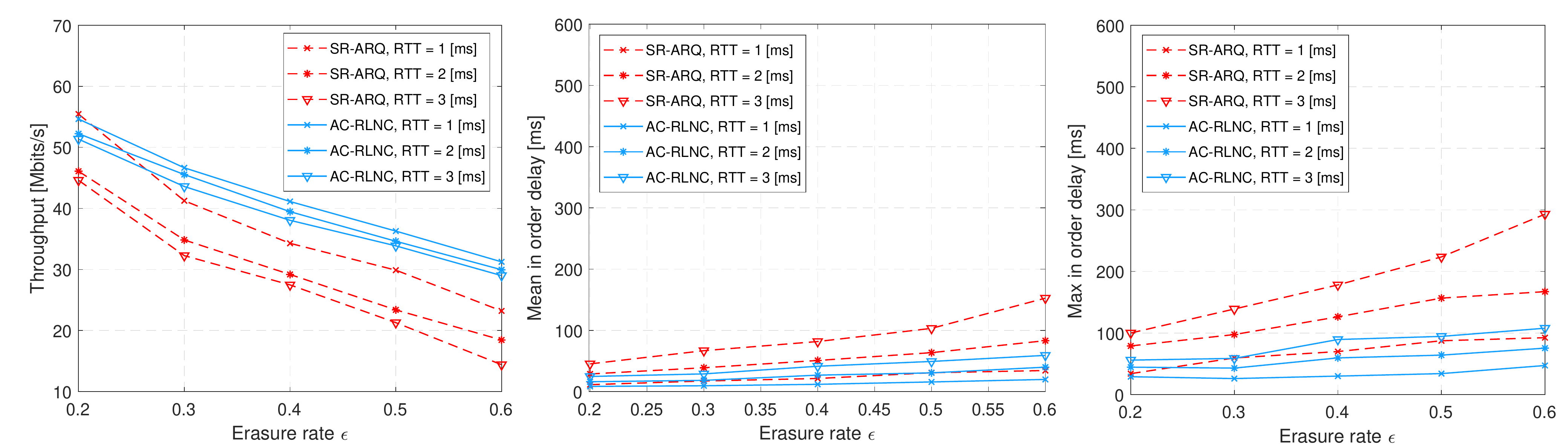}}
	\caption{Experimental study with Intel data. Throughput (left), mean in order delay (middle), and maximum in order delay (right).}
	\label{fig:ES}
\end{figure*}

We provide a performance comparison between SR-ARQ and AC-RLNC in terms of $\eta$, $\Dm$ and $\DM$, with respect to the rate $r$, for different values of RTT\footnote{slot durations can be adjusted based on the measurements from real systems}. We further assume that the retransmission criterion is $r>d$, i.e. $th=0$, i.e. the sender only tracks the average erasure rate per window and limits to $2k$ the maximum number of information packets allowed to overlap (i.e. $\ov=2k$). We compare the performance of different channel models: BEC and GE models.

{\bf Simulation for BEC.} In Fig. \ref{fig:BEC}, we evaluate the performance of AC-RLNC and contrast it with SR-ARQ's performance for a BEC in terms of $\eta$ (left), $\Dm$ (middle), and $\DM$ (right). The throughput gain of the AC-RLNC with respect to SR-ARQ is significant at high RTT values (nearly doubled at $\epsilon=0.4$). The throughput of both models is comparable when RTT is low because the need for adaptive FEC mechanism is eliminated. However, AC-RLNC always performs better. We can see significant gains (more than double at $\epsilon\geq 0.4$) in $\Dm$ and $\DM$ of AC-RLNC with respect to SR-ARQ. This is because the sender compares the DoFs received, $d$ (learnt via feedback) with the transmission rate, $r$ of the channel, adapts the FEC insertion rate. This enables a higher number of successful transmissions than the time-invariant streaming codes (Fig. \ref{fig:foobar}-a), and systematic codes with feedback (Fig. \ref{fig:foobar}-b). %As also demonstrated in Fig. \ref{fig:foobar}, for a given deadline constraint.
%In Sect. \ref{simulation} and especially in Fig. \ref{fig:BEC}, we demonstrate the performance with BEC of AC-RLNC with feedback given in Sect. \ref{algo} by simulations.

{\bf Simulation for GE channel.} In Fig. \ref{fig:GE}, we investigate and contrast the performance of the AC-RLNC with the performance of SR-ARQ for the GE channel where \ton{$q = \epsilon$ and} $s=0.3$, in terms of the $\eta$ (left), $\Dm$ (middle), and $\DM$ (right). While $\eta$ drops for both models because the channels are bursty, it is still possible to see the gain of using AC-RLNC. When we look at the delay performance, however, we can see that AC-RLNC can handle the burst erasures better than SR-ARQ in the sense that the gain is at least tripled for $\epsilon\geq 0.4$. Similarly, when we compare the $\DM$ performances, we can see that AC-RLNC is more stable in terms of the maximum in order delivery delay, such that the delay tail has a sub-Gaussian behavior.

\off{\ton{I am not sure if the following 2 paragraphs correspond to the same set of simulations.}
In Fig \ref{fig:DB}, for BEC we present the throughput, mean in order delay, and maximum in order delay bounds along with the numerical simulations, where we select $th=\Tilde{v}_e(i)$ with different values of $\lambda$.}

{\bf Validation with Intel traces.}
We next explain the implementation setup we consider in order to validate our algorithm in Sect. \ref{algo}, and the analytical bounds presented in Sect. \ref{bounds}. This is a controlled-congested setup considered by Intel, which is summarized in Fig. \ref{fig:Intel_Setup}. In this setup, there are three access points (AP1, AP2, AP3) and four clients, where one of the client is multi-user (Client 2). The Device Under Test (DUT) is the receiver, and there is one general interferer (using a signal generator). The dashed lines to the DUT are the point-to-point communication channels to be analyzed. The transmit powers can be adjusted to test several scenarios. The WiFi standards used are IEEE 802.11n, IEEE 802.11ax and IEEE 802.11ac. The WiFi traces are collected from all senders (clients and APs). The traces contain the time slot information of the feedback acknowledgements obtained from the receiver.
In Fig. \ref{fig:ES}, using the controlled-congested setup of Intel, we demonstrate the behavior of $\eta$ (left), $\Dm$ (middle), and $\DM$ (right). From the simulation results, we can note that the experimental study results have a good agreement with the simulation results using the GE channel model with the appropriate burst parameter.

\off{\ton{Shall we move the Intel setup to the simulation section?}
%%%%%
{\bf Controlled-congested setup implementation by Intel.}
We next explain the implementation setup we consider in order to validate our algorithm and the analytical bounds presented in Sect. \ref{bounds}. This is a controlled-congested setup considered by Intel, which is summarized in Fig. \ref{fig:Intel_Setup}. In this setup, there are three access points and four clients, where one of the client is multi-user. The Device Under Test (DUT) is the receiver, and there is one general interferer (using a signal generator). The dashed lines to the DUT are the point-to-point communication channels to be analyzed. The transmit powers can be adjusted to test several scenarios. The WiFi standards used are IEEE 802.11n, IEEE 802.11ax and IEEE 802.11ac. The WiFi traces are collected from all senders (clients and access points). The traces contain the time slot information of the packets transmitted, and the feedback acknowledgements obtained from the receiver. In the experimental study in Sect. \ref{simulation}, we validate this controlled-congested setup of Intel, in terms of the throughput, mean in order delay and maximum in order delay}
%%%%%%%%%%%%%%%%%%%%%%%%%%%%%%%%%%%%%%%%%%%%%%%%%%%%%%%%%%%%%%%
\subsection{Discussion}
%%%%%%%%%%%%%%%%%%%%%%%%%%%%%%%%%%%%%%%%%%%%%%%%%%%%%%%%%%%%%%%
The advantage of the AC-RLNC with feedback in terms of throughput and delay is even more evident when we have a GE channel instead of a BEC. This is because the adaptive causal network coding protocol can track the erasure pattern, and is more robust to burst erasures.

The causal and adaptive RLNC algorithm can be used for a variety of scenarios with different delay requirements. Since the algorithm is adaptive and robust to burst erasures and brings down the gap between $\DM$ and $\Dm$, it can be used to run applications with different delay sensitivities (e.g. file download $\Dm$, and in movie streaming concerning with minimizing $\DM$). This indicates the versatility of the proposed coding approach. Our analytical results and bounds (provided in Sect. \ref{bounds}) permit accounting accurately for delay when coding.

If the $\RTT$ is long, e.g., in the case of satellite, WiMAX, the insertion of FEC of $m$ retransmissions can be split during the window, instead of including the FEC at the end of the window. How to optimally split the insertion of FEC is an interesting extension of the current work.

%%%
\off{\subsection{Use Cases of the Algorithm}
The causal and adaptive RLNC algorithm can be used for a variety of scenarios with different delay requirements. Since the algorithm is adaptive and robust to burst erasures and brings down the gap between $\DM$ and $\Dm$, it can be used to run applications with different delay sensitivities (e.g. file download with the goal of minimizing the overall completion time that is proportional to $\Dm$, and in movie streaming with the goal of minimizing the maximum in order delay $\DM$). This indicates the versatility of the proposed coding approach. Our analytical results (provided in Sect. \ref{bounds}) permit accounting accurately for delay when coding.}

%%%%%%%%%%%%%%%%%%%%%%%%%%%%%%%%%%%%%%%%%%%%%%%%%%%%%%%%%%%%%%%
\section{Conclusions} \label{conc}
%%%%%%%%%%%%%%%%%%%%%%%%%%%%%%%%%%%%%%%%%%%%%%%%%%%%%%%%%%%%%%%
We proposed a causal and adaptive RLNC-based algorithm (i.e. AC-RLNC) for erasure channels. AC-RLNC can increase the throughput gains and reduce the in order delivery delay (not only the mean, but also the maximum) significantly. We derive the throughput and the mean and maximum in order delivery delay of AC-RLNC. The numerical simulations suggest that the adaptive model can predict the behavior in bursty environments and improve the gains even further. This is also consistent with the experimental validation of Intel traces. The proposed approach is a good starting point to demonstrate the gains that can be obtained via a causal and adaptive coding model. The throughput gains can be more than twice of SR-ARQ for memoryless channels, the gain in terms of the mean in order delivery delay is tripled for bursty channel, and the maximum in order delivery is more stable compared to SR-ARQ.

Future work includes the extension of the single-path model to multi-path, and optimize the packet scheduling. Extensions also include the study of more general mesh networks, where the interference and congestion are not negligible. These tradeoffs can be also exploited to see the fundamental limits of delay and throughput with hardware constraints from a practical point of view.

%%%%%%%%%%%%%%%%%%%%%%%%%%%%%%%%%%%%%%%%%%%%%%%%%%%%%%%%%%%%%%%
\section*{Acknowledgment}
%%%%%%%%%%%%%%%%%%%%%%%%%%%%%%%%%%%%%%%%%%%%%%%%%%%%%%%%%%%%%%%
We want to thank Shlomi Buganim for his contributions in the experimental study given in Sect. \ref{simulation}.

%%%%%%%%%%%%%%%%%%%%%%%%%%%%%%%%%%%%%%%%%%%%%%%%%%%%%%%%%%%%%%%
\appendix\label{App:adaptivecausalRLNC}
%%%%%%%%%%%%%%%%%%%%%%%%%%%%%%%%%%%%%%%%%%%%%%%%%%%%%%%%%%%%%%%
We provide a detailed description of the AC-RLNC algorithm realization in Fig. \ref{fig:foobar}-(d) for a point-to-point channel.
%%%%%%%%%%%%%%%%%%%%%%%%%%%%%%%%%%%%%%%%%%%%%%%%%%%%%%%%%%%%%%%

\begin{enumerate}[wide, labelwidth=!, labelindent=0pt]

    \item In Fig. \ref{fig:foobar}-(d) , we show a coding matrix using AC-RLNC with feedback. In each time slot $t$, the sender transmits a coded packet $c_t$ to the receiver. The receiver sends feedback (without any losses or delay) on each coded packet transmitted. For the $t^{\rm th}$ coded packet transmitted, the sender receives a reliable acknowledgment (e.g. ACK(t) or NACK(t)) after an RTT. For this particular example, $\RTT$ is $4$ time slots. %[0014]

    \item  In this example, at $t=1$, the sender transmits information packet $p_1$ to the receiver. The effective window size is $k=3$ (as $\RTT=k+1$). Since no feedback has been received from the receiver, and it is not the end of the effective window (EW),  %(e.g. the effective window does not end with $k$ new information packets),
    the sender includes packet $p_1$ into the effective window and generates a coded packet $c_1$ that includes a linear combination of the information packets included in the effective window (e.g. $p_1$), and transmits $c_1$ to the receiver. %\tof{The sender can then update the rate of DoF $d=m_d/a_d$.? It cannot because the sender does not have the feedback.}
    The sender may check to determine whether coded packet $c_1$ satisfies $DoF(c_t)\geq 2k$, and concludes that $c_1$ does not exceed the allowable DoF.%[0017]

    \item At $t=2$, the sender wants to transmit information packet $p_2$ to the receiver. Since no feedback has been received from the receiver, and the effective window does not end with $k$ information packets, the sender includes $p_2$ into the effective window and generates and transmits a coded packet $c_2$ that is a linear combination of the packets in the effective window (i.e. $p_1$ and $p_2$). The sender can then update the rate of DoF. It also checks whether $DoF(c_2)>2k$.%$c_2$ exceeds the allowable DoF (e.g. $DoF(c_t)>2k$). %[0018]

    \item At $t=3$, the sender wants to transmit information packet $p_3$. Since no feedback has been received from the receiver, and the effective window does not end with $k$ new information packets, the sender includes packet $p_3$ into the effective window and generates and sends a coded packet $c_3$ that is a linear combination of the information packets included in the effective window (i.e., $p_1$, $p_2$ and $p_3$). The sender then updates the rate of DoF and checks if $DoF(c_3)>2k$. %the DoF condition is satisfied. %[0019]

    \item At $t=4$, the sender wants to transmit information packet $p_4$. However, notwithstanding that no feedback has been received from the receiver, the sender checks and determines that the effective window ends with $k$ new information packets (i.e. $p_1$, $p_2$ and $p_3$). As a result, the sender generates and transmits a FEC packet to the receiver. For example, the FEC packet at $t=4$ can be a new linear combination of $p_1$, $p_2$ and $p_3$. It is also possible that the sender may transmit the same FEC a multiple times. We let the number of FECs be $m$, which can be specified based on the average erasure probability, or provided over the feedback channel. In other words, $m$ is a tunable parameter. In Fig. \ref{fig:foobar}-(d), $m=1$, which results in low throughput when the channel rate is high, and low in order delay when the rate is low. Hence, $m$ can be adaptively adjusted exploiting the value of the average erasure rate in the channel to achieve a desired delay-throughput tradeoff. %[0020] In the patent document there are lots of typos.

    \item At $t=4$, transmission of FEC packet is noted by the designation ``fec" in row $4$. The sender can then update the DoF added to $c_t$, and the rate of DoF. Note that the transmission of this FEC is initiated by the sender (e.g., the transmission of the FEC packet was not %requested or otherwise
    a result of feedback from the receiver). Furthermore, $p_4$ is not included in the effective window. The sender may check and determine that $c_4$ does not exceed the allowable DoF ($DoF(c_4)>2k$). %[0021]

    \item At $t=5$, the sender determines that it needs to transmit $p_4$. It also sees the acknowledgement of the receipt of $c_1$ (denoted by ACK(1) in the Feedback column at row $t=5$). Hence, the sender can remove a DoF (e.g., $p_1$) from the effective window. The effective window slides to the right (equivalent to sliding window). In this case, the sender determines that the effective window does not end with $k$ new information packets, and checks if the rate $r$ is higher than the DoF rate $d$, i.e. the threshold condition for retransmission (e.g., transmission of an FEC) is $th=0$. Hence, the sender can decide that the channel rate $r$ is sufficiently higher than the DoF rate $d$ (e.g. denoted by $(1-0/1)-0/1>0$ in the right-most column at row $t=5$). Note that $r=1-e/t$, where $e$ is the number of erasure packets and $t$ is the number of transmitted packets for which the sender has received acknowledgments. Having determined that $r-d\geq 0$, the sender includes $p_4$ into the effective window and generates a coded packet $c_5$ that includes a linear combination of the information packets in the effective window (i.e., $p_2$, $p_3$, and $p_4$), and transmits $c_5$ to the receiver. The sender then updates the rate of DoF. It may also check that $c_4$ does not exceed the allowable DoF ($DoF(c_4)>2k$).  %[0022]

    \item At $t=6$, the sender wants to transmit $p_5$. It also sees the acknowledgment of the receipt of $c_2$ (denoted by ACK(2) in the Feedback column at row $t=6$). Hence, the sender can remove a DoF (e.g., $p_2$) from the effective window. In this case, the sender determines that the effective window does not end with $k$ new information packets (but the effective window ends with new information packet $p_4$). The sender checks if $r-d\geq 0$. Having determined that $(1-0/2)-0/1>0$ (in the right-most column at row $t=6$), the sender includes $p_5$ into the effective window and generates and transmits a coded packet $c_6$ that includes a linear combination of the information packets included in the effective window (e.g., $p_3$, $p_4$ and $p_5$). The sender then updates the rate of DoF. It also checks that $c_6$ does not exceed the allowable DoF ($DoF(c_6)>2k$).   %[0023]

    \item At $t=7$, the sender wants to transmit $p_6$. The sender also sees a negative acknowledgment indicating the non-receipt of $c_3$ (denoted by NACK(3) in the Feedback column at row $t=7$). The non-receipt of $c_3$ is denoted by the x'ed out dot in row $3$. Upon the reception of NACK, the sender increments a count of the number of erasures (denoted as $e$). For example, since this is the first erasure, the count of $e$ is incremented to the value of one. The sender then updates the missing DoF to decode $c_7$ according to the number of NACKs received. In this instance, the sender updates the missing DoF to decode $c_7$ to a value of one ($1$) since this is the first NACK. The sender then checks if the retransmission condition is satisfied. Since $(1-1/3)-1/1<0$, it generates and transmits an FEC packet $c_7$ to the receiver. In this instance, the FEC packet $c_7$ is a new linear combination of $p_3$, $p_4$ and $p_5$. Note that the threshold condition $(1-1/3)-1/2<0$ noted in the right-most column at row $t=7$ indicates the state of the threshold condition subsequent to transmission of $c_7$. The transmission of this FEC is a result of feedback by the receiver (which is noted as ``fb-fec" in row $t=7$). The sender then updates the DoF added to $c_7$. The sender determines that the effective window does not end with $k$ new information packets (e.g. effective window ends with new information packets $p_4$ and $p_5$), and then updates the rate of DoF. It also checks that $c_7$ does not exceed the allowable DoF ($DoF(c_7)>2k$). %[0024]

    \item At $t=8$, the sender determines that it still needs to send $p_3$ to the receiver. The sender also sees a negative acknowledgment indicating the non-receipt of $c_4$ (denoted by NACK(4) in the Feedback column at row $t=8$). Upon the receipt of NACK, the sender increments a count of $e$. In this instance, since this is the second erasure, the count of $e$ is incremented by one to a value of two (i.e. $e=1+1$). The sender then updates the missing DoF to decode $c_7$. After it checks that the retransmission threshold is not satisfied, the sender generates and transmits an FEC packet $c_8$ to the receiver. The FEC packet $c_8$ is a new linear combination of $p_3$, $p_4$ and $p_5$. The transmission of this FEC is a result of feedback by the receiver (which is noted as ``fb-fec" in row $t=8$). The sender then updates the DoF added to $c_8$, and determines that the effective window does not end with $k$ new information packets (effective window ends with new information packets $p_4$ and $p_5$). The sender then updates the rate of DoF. It may also check that $c_8$ does not exceed the allowable DoF ($DoF(c_8)>2k$). %[0025]

    \item At $t=9$, the sender determines that it still needs to send $p_3$ to the receiver. The sender also sees an acknowledgment indicating the receipt of $c_5$ (denoted by ACK(5) in the Feedback column at row $t=9$). Since sender sees an acknowledgment, and determines that the effective window does not end with $k$ new information packets (effective window ends with new information packets $p_3$, $p_4$ and $p_5$), the sender checks if $r-d\geq 0$. Having determined that $(1-2/5)-1/3>0$ (as shown in the right-most column at row $t=9$), the sender includes packet $p_6$ into the effective window, and generates and transmit a coded packet $c_9$ that includes a linear combination of the information packets included in the effective window (i.e., $p_3$, $p_4$, $p_5$ and $p_6$). The sender then updates the rate of DoF. It may also check that $c_9$ does not exceed the allowable DoF ($DoF(c_9)>2k$). %[0026]

    \item \mm{At $t=10$, the sender determines that it needs to send $p_3$ to the receiver. The sender also sees an acknowledgment indicating the receipt of $c_6$ (denoted by ACK(6)). Since the sender sees an acknowledgment, it determines that the effective window ends with information packets $p_3$, $p_4$, $p_5$, and $p_6$. As a result, the sender generates and transmits a FEC packet $c_{10}$. The sender then updates the added DoF (i.e., $ad=3+1$). Now, since $(1-2/6)-1/4>0$, the sender includes packet $p_7$ into the effective window. Then it generates a coded packet $c_{11}$ to transmit at time slot $11$ that includes a linear combination of the information packets in the effective window (i.e., $p_3$, $p_4$, $p_5$, $p_6$ and $p_7$). It may also check that $c_{11}$ does not exceed the allowable DoF ($DoF(c_{11})>2k$).}

    \off{item At $t=10$, the sender determines that it needs to send $p_3$ to the receiver. The sender also sees a negative acknowledgment indicating the non-receipt of $c_6$ (denoted by NACK(6) in the Feedback column at row $t=10$). Hence, the sender increments the number of erasures by one to a value of three (i.e., $e=2+1$). The sender then updates the missing DoF to decode $c_9$. Having determined that $(1-3/6)-2/3<0$ (as shown in the right-most column at row $t=10$), it generates and transmits an FEC packet $c_{10}$ to the receiver. The FEC packet $c_{10}$ is a new linear combination of information packets $p_3$, $p_4$, $p_5$ and $p_6$. Transmission of this FEC is a result of feedback by the receiver (which is noted as ``fb-fec" in row $t=10$). The sender then updates the DoF added to $c_{10}$. The sender then checks that the effective window ends with new information packets $p_3$, $p_4$, $p_5$, and $p_6$. As a result, the sender generates and transmits a FEC packet (this occurs during the next time slot i.e. $t+1=11$). In this case, since the most recently transmitted coded packet is $c_{10}$ (at time $t=10$), the FEC will be a coded packet $c_{11}$ to be transmitted (at time $t=11$). The sender then updates the rate of DoF. It also checks that $c_{10}$ does not exceed the allowable DoF ($DoF(c_{10})>2k$).} %[0027]

   \item \mm{At $t=11$, the sender transmits $c_{11}$. However, according to the acknowledgment indicating the receipt of $c_7$ (denoted by ACK(7)), the sender removes DoF (e.g., information packets $p_3$, $p_4$ and $p_5$) from the effective window. The sender then updates the rate of DoF.} %It may also check that $c_{11}$ does not exceed the allowable DoF ($DoF(c_{11})>2k$).} %[0028]

    \item \mm{At $t=12$, the sender sees an acknowledgment indicating the receipt of $c_7$ (denoted by ACK(7) in the Feedback column at row $t=12$). Since the sender sees an acknowledgment, and determines that the effective window does not end with $k$ new information packets (effective window ends with new information packet $p_7$), the sender checks if $r-d\geq 0$. Having determined that $(1-3/8)-0/1>0$ (as shown in the right-most column at row $t=12$), the sender includes packet $p_8$ into the effective window. Then the sender generates and transmits a coded packet $c_{12}$. This coded packet includes a linear combination of the information packets available in the effective window (i.e., $p_6$, $p_7$ and $p_8$). It may also check that $c_{12}$ does not exceed the allowable DoF ($DoF(c_12)>2k$).} %[0029]

   \item The rest of the example follows using similar steps. The sender continues sending information or coded packets during time slots $t\in[13,\,28]$. Hence, the sender can adaptively adjust its transmission rate according to the process described above.%[0030]
  \end{enumerate}

%%%%%%%%%%%%%%%%%%%%%%%%%%%%%%%%%%%%%%%%%%%%%%%%%%%%%%%%%%%%%%%
\bibliographystyle{IEEEtran}
\bibliography{references}
%%%%%%%%%%%%%%%%%%%%%%%%%%%%%%%%%%%%%%%%%%%%%%%%%%%%%%%%%%%%%%%
\end{document}